\documentclass{article}

\usepackage[main, preprint]{neurips_2026}
\usepackage{wrapfig,lipsum,booktabs}
\usepackage[utf8]{inputenc}
\usepackage{natbib}
\usepackage[T1]{fontenc}
\usepackage{hyperref}
\usepackage{url}
\usepackage{booktabs}
\usepackage{amsfonts}
\usepackage{nicefrac}
\usepackage{microtype}
\usepackage{xcolor}
\usepackage{graphicx}
\usepackage{subcaption}

\usepackage{amsmath}
\usepackage{amssymb}
\usepackage{mathtools}
\usepackage{amsthm}
\usepackage{enumitem}
\usepackage{multirow}

\usepackage{algorithm}
\usepackage{algorithmic}
\usepackage[capitalize,noabbrev]{cleveref}

\theoremstyle{plain}
\newtheorem{theorem}{Theorem}[section]
\newtheorem{proposition}[theorem]{Proposition}
\newtheorem{lemma}[theorem]{Lemma}
\newtheorem{corollary}[theorem]{Corollary}
\theoremstyle{definition}
\newtheorem{definition}[theorem]{Definition}

\theoremstyle{remark}
\newtheorem{remark}[theorem]{Remark}

\usepackage[disable,textsize=tiny]{todonotes}

\usepackage{xspace}
\newcommand{\cX}{\mathcal{X}}           % Token space
\newcommand{\cV}{\mathcal{V}}           % Vocabulary
\newcommand{\cA}{\mathcal{A}}           % Attack set
\newcommand{\cC}{\mathcal{C}}           % Computational cost
\newcommand{\cL}{\mathcal{L}}           % Loss function
\renewcommand{\cC}{\mathcal{C}}           % Computational cost
\renewcommand{\cL}{\mathcal{L}}           % Loss function
\newcommand{\cS}{\mathcal{S}}           % Sensitive token set

\newcommand{\Ent}[1]{H\left(#1\right)}                    % Entropy
\newcommand{\MI}[2]{I\left(#1; #2\right)}                 % Mutual information
\newcommand{\KL}[2]{\text{KL}\left(#1 \| #2\right)}       % KL divergence
\newcommand{\Hcond}[2]{H\left(#1 | #2\right)}             % Conditional entropy

\newcommand{\E}{\mathbb{E}}              % Expectation
\newcommand{\Prob}{\mathbb{P}}           % Probability
\newcommand{\R}{\mathbb{R}}              % Real numbers
\newcommand{\Emb}{E}                     % Embedding layer
\newcommand{\LN}{\operatorname{LN}} 

\newcommand{\PIA}{PIA\xspace}             % Prompt Inversion Attack
\crefname{theorem}{Theorem}{Theorems}
\crefname{lemma}{Lemma}{Lemmas}
\crefname{corollary}{Corollary}{Corollaries}
\crefname{definition}{Definition}{Definitions}
\crefname{algorithm}{Algorithm}{Algorithms}
\crefname{figure}{Figure}{Figures}
\crefname{table}{Table}{Tables}
\crefname{equation}{Eq.}{Eqs.}
\crefname{section}{Section}{Sections}

\newcommand{\norm}[1]{\left\| #1 \right\|}
\newcommand{\abs}[1]{\left| #1 \right|}

\newcommand{\Acc}{\text{Acc}}
\newcommand{\RETURN}{\STATE \textbf{return}~}

\renewcommand{\Acc}{\text{Acc}}

\usepackage{xcolor}
\title{Defense Against Prompt Inversion Attacks: An Information-Theoretic
       Approach for LLM Collaborative Inference}

\author{
  Sayedeh Leila Noorbakhsh \quad Hossein Khalili \quad Nader Sehatbakhsh \\
  University of California, Los Angeles \\
  \texttt{\{lnoorbakhsh, hkhalili, nsehat\}@ucla.edu}
}

\begin{document}

\maketitle

%------------------------------------------------------------------
% ABSTRACT
%------------------------------------------------------------------
% sections/abstract.tex
% Abstract of the paper

\begin{abstract}
% We propose an information-theoretic defense framework against prompt inversion attacks in Large Language Model (LLM) collaborative inference systems. In such systems, resource-constrained IoT devices offload computation to cloud servers by transmitting intermediate activations, which exposes user prompts to reconstruction attacks. Our defense learns privacy-preserving representations that minimize information leakage about the input prompt while preserving utility and satisfying computational constraints. We derive theoretical guarantees on reconstruction error, establish inherent privacy-utility tradeoffs, and provide token-level accuracy bounds. Experimental results demonstrate our approach achieves superior privacy-utility-latency tradeoffs compared to existing defenses.

Collaborative edge-cloud inference enables resource-constrained devices to leverage large language models (LLMs) by offloading partial computation to cloud servers. However, transmitting intermediate activations exposes sensitive user prompts to \emph{prompt inversion attacks}, where an adversary reconstructs the original input from shared representations. Existing defenses rely largely on heuristic perturbations or empirical tuning, offering limited theoretical understanding of privacy leakage and its interaction with utility and latency constraints.
We propose an information-theoretic defense framework for prompt inversion in collaborative LLM inference. Our approach learns privacy-preserving representations by explicitly minimizing the mutual information between intermediate activations and the input prompt while maintaining task utility under computational constraints. We derive theoretical guarantees on prompt reconstruction error, characterize fundamental privacy–utility tradeoffs, and establish token-level accuracy bounds for downstream inference. We then propose a novel defense based on privacy adapters implemented via 
low-dimensional information bottlenecks. Extensive experiments across multiple settings demonstrate that our method achieves superior privacy–utility–latency tradeoffs compared to existing defenses (up to 35\% reduction in attack success), providing a principled foundation for private and efficient collaborative LLM inference.
\end{abstract}

%------------------------------------------------------------------
% INTRODUCTION
%------------------------------------------------------------------
% sections/introduction.tex
% Introduction (compressed for 9-page limit)

\section{Introduction}
\label{sec:introduction}

Large Language Models (LLMs) have demonstrated remarkable capabilities across natural language processing tasks, but their substantial computational and memory requirements pose significant challenges for deployment on resource-constrained platforms such as mobile devices, IoT, and embedded edge systems~\citep{sheng2023flexgen,qu2025mobile,liu2024mobilellm}. Collaborative inference~\citep{borzunov2023distributed,wang2024cloud,kang2017neurosurgeon,zhang2024edgeshard,miao2025towards} addresses this by executing initial layers on-device and offloading intermediate activations to a cloud server for completion---but this exposes substantial information about the input prompt. Adversaries can mount \emph{prompt inversion attacks} using learned classifiers~\citep{luo2025prompt} or constrained optimization~\citep{qu2025prompt} to reconstruct sensitive user inputs. ~\citet{dong2025depth} show that token identity persists through \emph{all} transformer layers, making the defense problem fundamentally challenging~\citep{lu2024position,chien2023enc2}.

Existing defenses against prompt inversion rely largely on heuristics: noise injection~\citep{mai2023split}, dimensionality reduction~\citep{dong2025depth}, ad hoc architectural modifications~\citep{chen2024unveiling}, which can empirically reduce reconstruction accuracy but lack principled guarantees and offer little insight into the inherent privacy-utility tradeoff. We instead take an information-theoretic approach: we model privacy leakage as the mutual information between the prompt and transmitted activations, and formulate the defense as a minimax optimization in which the defender minimizes leakage while a learned adversary maximizes reconstruction, subject to a utility constraint on next-token prediction. We achieve this objective via \emph{privacy adapters}: lightweight information-bottleneck modules inserted between frozen transformer layers on the device side that compress activations through a low-dimensional subspace. Crucially, existing residual perturbation defenses provably cannot reduce mutual information; instead, our bottleneck design forces lossy compression and enables genuine privacy-utility tradeoffs.

A second insight motivates our evaluation: privacy protection in collaborative LLM inference need not be uniform. We introduce a \emph{sensitive token analysis} that separately evaluates the attacker's recovery of domain-specific tokens (medical terms, legal entities, identifiers) versus common structural tokens (stop words, punctuation). Privacy adapters disproportionately protect sensitive tokens, reducing their recovery by up to 35 percentage points relative to common tokens, while preserving structural information for downstream processing. The practical implication is significant: an attacker who recovers ``What information do we have for [PROTECTED] with [PROTECTED]?'' learns far less than one who recovers the full query.

While prior information-theoretic privacy work has targeted \emph{classification} settings~\citep{jaiswal2020invariant, osia2020hybrid, singh2021disco, singh2023posthoc, noorbakhsh2024inf2guard, zhang2025learning}, to the best of our knowledge, this is the first information-theoretic privacy defense for collaborative LLM inference under prompt inversion attacks. A key challenge in our setting is \textit{token identity} — the very information we aim to protect — is also essential to the generative task itself. Our framework navigates this tension through our novel privacy adapters controlled by two complementary factors: the tradeoff parameter $\lambda$ and the bottleneck dimension $r$. 

We validate the approach on LLaMA-2-7B/13B and Mistral-7B across three domains (medical, legal, airline reviews), against both classification-based~\citep{luo2025prompt} and optimization-based~\citep{qu2025prompt} attacks. Our method reduces the attack success rate by more than 30\% while incurring only $<$9\% latency overhead. Figure~\ref{fig:overview} shows the system.

\paragraph{Contributions.}
\begin{itemize}[itemsep=0em,topsep=2pt]
    \item An information-theoretic minimax framework for prompt inversion defense, with a learned adversary that directly estimates mutual information via cross-entropy, closing the theory-implementation gap left by prior surrogate-loss approaches.
    \item We present \emph{Privacy adapters}: information-bottleneck modules that compress activations through a low-dimensional subspace, supported by theoretical bounds on reconstruction error as a function of bottleneck dimension and a formal proof that injective residual-style defenses (covering LoRA-style adapters and learned perturbations) cannot reduce MI.
    \item A \emph{sensitive token analysis} methodology that distinguishes domain-specific from structural token recovery, providing a more meaningful privacy measure than uniform token accuracy.
    \item Comprehensive experiments across three LLMs, three domains, and three attack strategies, showing 15--35\% reductions in sensitive token recovery with $<$9\% latency overhead.
\end{itemize}

\begin{figure}
    \centering
    \includegraphics[width=.9\linewidth]{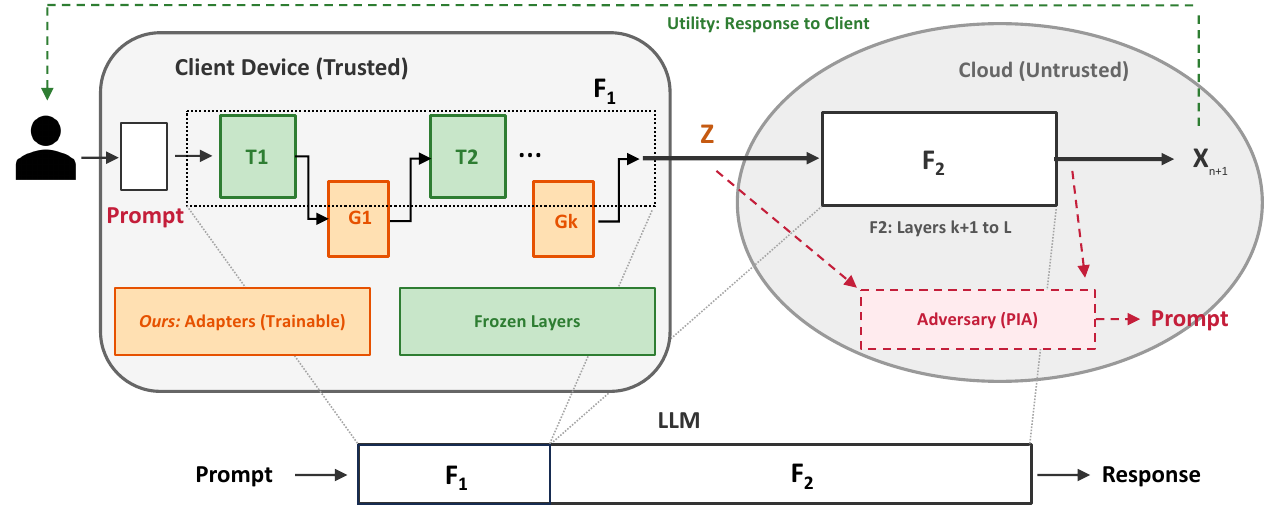}
    \caption{Privacy-adapter--based collaborative inference. Lightweight trainable adapters are inserted between frozen transformer layers on the client to suppress prompt leakage before offloading intermediate activations ($Z$) to the cloud. The adapters are optimized via a minimax objective balancing privacy protection and downstream utility under resource constraints.}
    \label{fig:overview}
\end{figure}

%------------------------------------------------------------------
% PROBLEM SETUP
%------------------------------------------------------------------
% sections/problem_setup.tex
% Problem Setup (compressed for 9-page limit)

\section{Problem Setup}
\label{sec:problem_setup}

\subsection{System Model}
\label{subsec:system_model}

We consider a two-party collaborative inference framework~\citep{qu2025mobile} with a trusted resource-constrained \emph{client} and an untrusted \emph{cloud server}. The client executes the embedding layer $\Emb$ and the first $k$ transformer layers $\{T_1, \ldots, T_k\}$, augmented with privacy adapters $\{G_1, \ldots, G_k\}$ that transform activations before transmission. The cloud executes the remaining layers $\{T_{k+1}, \ldots, T_L\}$ followed by the language modeling head $W_{\text{LM}} \in \R^{|\cV| \times d}$ (the overview of this system is shown in Figure~\ref{fig:overview}). Notation is summarized in Appendix~\ref{app:notation} (Table~\ref{tab:notation}).

\paragraph{Inference workflow.}
Given an input prompt $x = [x_1, \ldots, x_n] \in \cX^n$, the client interleaves transformer layers with privacy adapters:
\begin{equation}
\label{eq:interleaved}
    \mathbf{h}^{(0)} = \Emb(x), \qquad \mathbf{h}^{(i)} = G_i(T_i(\mathbf{h}^{(i-1)})) \quad \text{for } i = 1, \ldots, k,
\end{equation}
where $G_i: \R^{n \times d} \to \R^{n \times d}$ is an information-bottleneck module compressing activations through a low-dimensional subspace of dimension $r \ll d$ (full architecture in Section~\ref{sec:method}). Crucially, $G_i$ applies \emph{no residual connection}---the output replaces, rather than perturbs, the layer activation.This is information-theoretically necessary: for deterministic additive perturbations $\boldsymbol{\delta}$ in the typical regime of LoRA-style adapters and learned residuals (i.e., satisfying $L<1$ Lipschitz), $I(\mathbf{h} + \boldsymbol{\delta}(\mathbf{h}); x) = I(\mathbf{h}; x)$ since the mapping $\mathbf{h} \mapsto \mathbf{h}+\boldsymbol{\delta}(\mathbf{h})$ is injective (formalized in Theorem~\ref{thm:residual}). The bottleneck forces lossy compression, enabling genuine MI reduction. The protected activation $A = \mathbf{h}^{(k)} \in \R^{n \times d}$ is transmitted to the cloud, which produces $P(y \mid A) = \text{softmax}(W_{\text{LM}} \cdot T_L(\cdots T_{k+1}(A)))$.

\subsection{Threat Model}
\label{subsec:threat_model}

\paragraph{Adversary.}
We assume an honest-but-curious cloud server attempting to reconstruct the user's prompt from the protected activation $A$ (using privacy adapters discussed above). Such prompts may contain sensitive content---medical queries, legal proceedings, personal communications, proprietary documents. The adversary operates under the strongest white-box setting: full access to $A$, all cloud parameters $\{T_{k+1}, \ldots, T_L, W_{\text{LM}}\}$, and complete knowledge of the privacy adapter weights $\{G_1, \ldots, G_k\}$. The adversary may train auxiliary reconstruction models, including state-of-the-art classification-based~\citep{luo2025prompt} and optimization-based~\citep{qu2025prompt} prompt inversion attacks, to minimize token-level reconstruction error.

\paragraph{Defense goal.}
We seek privacy adapters $\{G_1, \ldots, G_k\}$ that simultaneously \emph{(i)} minimize information leakage $I(A; x)$, formalized as a minimax game with a learned adversary $q_\Psi$ (Section~\ref{sec:method}); \emph{(ii)} preserve next-token prediction utility, measured by perplexity; and \emph{(iii)} respect client-side latency and parameter budgets ($\ll 1\%$ of model parameters).

%------------------------------------------------------------------
% METHOD: DEFENSE FORMULATION
%------------------------------------------------------------------
\section{Method: Privacy-Preserving Collaborative Inference}
\label{sec:method}

%We formulate the defense as a minimax optimization that jointly addresses three goals: (i) minimizing privacy leakage from transmitted activations, (ii) preserving downstream language modeling utility, and (iii) respecting on-device computational constraints. We then realize this objective through lightweight \emph{privacy adapters} based on an information bottleneck.

\subsection{Formalizing Defense Objectives}
\label{subsec:objective}

\paragraph{Privacy via mutual information minimization.} 
Let $A = (G_k \circ T_k \circ \cdots \circ G_1 \circ T_1)(\Emb(x))$ denote the protected activation transmitted to the cloud. We seek to minimize the mutual information between $A$ and the input prompt $x$:
\begin{equation}
\label{eq:privacy_objective}
\min_{\{G_i\}} \MI{A}{x}.
\end{equation}

Small MI implies that even an optimal adversary cannot reliably reconstruct $x$ from $A$~\citep{noorbakhsh2024inf2guard, arevalo2024task}. Since $\MI{A}{x}$ is intractable, we use the variational upper bound vCLUB~\citep{cheng2020club} 
(Lemma~\ref{lem:vclub}; statement and proof in Appendix~\ref{app:vb}) and parameterize the variational distribution as a per-position token classifier $q_\Psi(x|A) = \prod_{j=1}^n q_\Psi(x_j|A_j)$, consistent with the threat model of~\cite{luo2025prompt}. Tightening this bound is equivalent to training $\Psi$ to minimize a per-position cross-entropy loss (Appendix~\ref{app:tightening}), yielding the minimax formulation:
\begin{equation}
\label{eq:privacy_minimax}
\min_{\{G_i\}} \max_\Psi \; \E_{p(x)}\!\left[\sum_{j=1}^{n} \log q_\Psi(x_j \mid A_j)\right].
\end{equation}

% Unlike cosine-similarity surrogates, this learned-adversary formulation requires no distributional assumptions on $A$, has unambiguous gradient direction, and adapts to the current adapter state during training. The per-position factorization is conservative: it yields a looser MI bound than a sequence-level adversary, so any reduction in our bound also reduces the true MI.

Unlike cosine-similarity surrogates, this learned-adversary formulation requires no distributional assumptions on $A$, has unambiguous gradient direction, and adapts to the current adapter state during training. By the chain rule of mutual information (Lemma~\ref{lem:per_position_bound}, Appendix~\ref{app:per_position}), $\MI{A}{x} \le \sum_{j=1}^n \MI{A}{x_j}$, so the per-position objective targets a valid upper bound on the joint MI; reducing $\sum_j \MI{A_j}{x_j}$ in turn directly tightens the success bound for per-position attackers~\citep{luo2025prompt, qu2025prompt} via Theorem~\ref{thm:token_acc} (see Remark~\ref{rem:surrogate}).

\paragraph{Utility via lower bound on $\MI{A}{y}$.}
Let $y = x_{n+1}$. Maximizing $\MI{A}{y}$ retains predictive information for next-token prediction. Using the variational lower bound (Lemma~\ref{lem:utility_lb}, Appendix~\ref{app:lem}) and noting that $p_{\text{LM}}$ is the fixed output distribution of the frozen cloud, the utility objective reduces to:
\begin{equation}
\label{eq:utility_objective}
\min_{\{G_i\}} \mathcal{L}_{\text{NTP}}(A) := \E_{p(x)}\!\left[-\log p_{\text{LM}}\!\left(x_{n+1} \mid T_L(\cdots T_{k+1}(A))\right)\right].
\end{equation}
Crucially, $\mathcal{L}_{\text{NTP}}$ is computed by forwarding $A$ through \emph{all} frozen cloud layers and the LM head, ensuring that adapters learn representations compatible with the downstream pipeline rather than merely mimicking original activations.

\paragraph{Latency constraint.}
For real-time deployment on resource-constrained devices, we impose $\cC(\{G_i\}_{i=1}^k) \leq C_{\max}$, where $\cC$ measures adapter-induced overhead. Table~\ref{tab:constraints} (Appendix~\ref{app:constraints}) summarizes the system constraints.

\paragraph{Combined objective.}
Combining privacy, utility, and latency yields:
\begin{equation}
\label{eq:combined_objective}
\boxed{
\min_{\{G_i\}} \max_\Psi \; \lambda \cdot \!\left(-\E_{p(x)}\!\left[\sum_{j=1}^{n}\log q_\Psi(x_j \mid A_j)\right]\right) + (1-\lambda) \cdot \mathcal{L}_{\text{NTP}}(A),
}
\end{equation}
where $\lambda \in [0,1]$ controls the privacy-utility tradeoff, subject to the latency constraint above.

\subsection{Privacy Adapter Architecture}
\label{subsec:adapter}

We realize $\{G_i\}$ as lightweight modules inserted after each device-side transformer layer. All pre-trained transformer layers remain frozen, avoiding expensive fine-tuning while preserving pre-trained knowledge.

\paragraph{Information bottleneck, no residual.}
Each adapter $G_i$ implements a compression-expansion through a low-dimensional subspace, \emph{without} a residual connection:
\begin{equation}
\label{eq:adapter_structure}
G_i(\mathbf{h}) = \LN\!\left(W_{\text{up}}^{(i)} \cdot \sigma\!\left(\LN\!\left(W_{\text{down}}^{(i)} \cdot \mathbf{h}\right)\right)\right),
\end{equation}
where $W_{\text{down}}^{(i)} \in \mathbb{R}^{r \times d}$, $W_{\text{up}}^{(i)} \in \mathbb{R}^{d \times r}$, $\sigma$ is GELU, and $r \ll d$ is the bottleneck dimension. Unlike LoRA-style adapters~\citep{hu2022lora} that compute $\mathbf{h} + G_i(\mathbf{h})$, our design \emph{replaces} the activation: $\mathbf{h}^{(i)} = G_i(T_i(\mathbf{h}^{(i-1)}))$.

The absence of a residual is essential. For deterministic perturbations $\boldsymbol{\delta}(\mathbf{h})$ in the operating regime of standard residual defenses (LoRA-style adapters with small initialization, learned additive perturbations with bounded weights), the mapping $\mathbf{h} \mapsto \mathbf{h} + \boldsymbol{\delta}(\mathbf{h})$ is injective, so $I(\mathbf{h} + \boldsymbol{\delta}(\mathbf{h}); x) = I(\mathbf{h}; x)$ by the data processing inequality (Theorem~\ref{thm:residual}). Thus, such residual-based defenses cannot reduce MI---a sufficiently powerful attacker recovers the original $\mathbf{h}$ and extracts token identity. We verify this empirically in Section~\ref{sec:experiments}: residual adapters yield $>$85\% token recovery under a fresh attacker. In contrast, the bottleneck $W_{\text{down}}^{(i)}: \mathbb{R}^d \to \mathbb{R}^r$ with $r \ll d$ enforces lossy compression, with $r$ as an explicit knob on information capacity.

\paragraph{Parameter efficiency.}
The adapters introduce only $2 \cdot k \cdot r \cdot d$ trainable parameters, typically $<0.3\%$ of the full model. For LLaMA-2-7B with $k=4, r=512$, this is $\sim$17M parameters out of 6.7B. Interleaving adapters across $k$ device layers enables layer-wise selective suppression of token identity.

\subsection{Optimization}
\label{subsec:optimization}

We optimize Eq.~\eqref{eq:combined_objective} via alternating updates (Algorithm~\ref{alg:training}, Appendix~\ref{alg:training}):
(i) \textbf{Adversary step} ($S$ steps): update $\Psi$ to minimize per-position cross-entropy on defended activations, tightening the MI upper bound;
(ii) \textbf{Adapter step} (1 step): with $\Psi$ frozen, update $\{G_i\}$ to minimize the combined privacy-utility loss.
All transformer layers $T_i$ remain frozen throughout.

%------------------------------------------------------------------
% THEORETICAL ANALYSIS
%------------------------------------------------------------------
% sections/theory.tex
% Theoretical analysis of the defense (compressed for 9-page limit)

\section{Theoretical Analysis}
\label{sec:theory}

We provide theoretical guarantees on three key questions: (i) does minimizing $I(A;x)$ provably limit reconstruction? (ii) why must the architecture avoid residual connections? and (iii) how does the bottleneck dimension control privacy at the token level? Full statements, additional results (bottleneck MI bound, utility–privacy tradeoff, token-level accuracy), and all proofs are deferred to Appendix~\ref{app:proofs}.

\paragraph{Privacy guarantee via Fano's inequality.}
Let $\cA_{\PIA} = \{\cA: \R^{n \times d} \to \cX^n\}$ denote the set of all prompt inversion attacks. Applying Fano's inequality~\citep{cover2006elements} to the Markov chain $x \to A \to \hat{x}$ yields a universal upper bound on adversary success.

\begin{theorem}[Privacy Leakage Bound]
\label{thm:privacy_bound}
For a random prompt $x$ with protected activation $A = (G_k \circ T_k \circ \cdots \circ G_1 \circ T_1)(\Emb(x))$ and any attack $\cA \in \cA_{\PIA}$,
\begin{equation}
\Prob\left[\cA(A) = x\right] \leq \frac{\MI{x}{A} + 1}{\log |\cX^n|}.
\end{equation}
\end{theorem}

\textit{Proof.} See Appendix~\ref{app:proof-tradeoff}. \hfill $\square$

This bound is universal: it holds for \emph{any} attack, including both classification-based~\citep{luo2025prompt} and optimization-based~\citep{qu2025prompt} approaches. Reducing the per-position MI $\sum_j I(A_j; x_j)$ via our minimax objective tightens the bound for per-position attackers (Theorem~\ref{thm:token_acc}) and, via Lemma~\ref{lem:per_position_bound}, also constrains the joint $I(A; x)$ through the chain-rule inequality.. A token-level analog (Theorem~\ref{thm:token_acc} in Appendix~\ref{app:proof-token-bound}) gives $\E[\Acc_{\text{token}}] \leq \frac{1}{n}\sum_i \frac{I(A;x_i)+1}{\log|\cV|}$, which provides the basis for our sensitive vs.\ common token analysis below.

\paragraph{Why residual connections cannot provide privacy.}
A central architectural choice in our design is the absence of a residual connection in $G_i$. The following result shows that this is not an empirical preference but a theoretical necessity: residual-style defenses in the regime of standard adapter constructions fail to reduce mutual information.

% \begin{theorem}[Residual MI Invariance]
% \label{thm:residual}
% Let $\boldsymbol{\delta}: \R^d \to \R^d$ be any deterministic function. For the residual defense $A = \mathbf{h} + \boldsymbol{\delta}(\mathbf{h})$ where $\mathbf{h} = (T_k \circ \cdots \circ T_1)(\Emb(x))$,
% \begin{equation}
% I(\mathbf{h} + \boldsymbol{\delta}(\mathbf{h}); x) = I(\mathbf{h}; x).
% \end{equation}
% \end{theorem}

% \textit{Proof sketch.} The map $\phi: \mathbf{h} \mapsto \mathbf{h} + \boldsymbol{\delta}(\mathbf{h})$ is deterministic and injective (under bounded gradient), so the data processing inequality holds with equality. Since $\phi(\mathbf{h})$ contains $\mathbf{h}$ as a recoverable component, an adversary loses no information. Full proof in Appendix~\ref{app:proofs}. \hfill $\square$

\begin{theorem}[Residual MI Invariance under Injectivity]
\label{thm:residual}
Let $\boldsymbol{\delta}: \R^d \to \R^d$ be deterministic and let $\phi(\mathbf{h}) = \mathbf{h} + \boldsymbol{\delta}(\mathbf{h})$. If $\phi$ is injective, then
\begin{equation}
I(\mathbf{h} + \boldsymbol{\delta}(\mathbf{h}); x) = I(\mathbf{h}; x).
\end{equation}
A sufficient condition is that $\boldsymbol{\delta}$ is $L$-Lipschitz with $L<1$, in which case $\phi$ is a bi-Lipschitz homeomorphism on $\R^d$. Proof and counterexamples in Appendix~\ref{app:proof-residual}.
\end{theorem}

% This explains a key empirical finding (Section~\ref{sec:experiments}): residual adapters drop the \emph{training adversary} to $\sim$50\% accuracy via co-adaptation, but a \emph{fresh attacker} trained from scratch on the same protected activations recovers $>$85\% token accuracy. The information was never removed---only hidden from one specific adversary. Stochastic perturbations $\boldsymbol{\delta} \perp \mathbf{h}$ \emph{can} reduce MI but introduce non-determinism at inference time (unreproducible outputs); our bottleneck approach achieves MI reduction deterministically.

The injectivity hypothesis covers residual defenses deployed in practice---LoRA-style adapters, learned additive perturbations with bounded weights, and similar constructions satisfy $L<1$ Lipschitz and fall under Theorem~\ref{thm:residual}. This explains a key empirical finding (Section~\ref{sec:experiments}): residual adapters drop the \emph{training adversary} to $\sim$50\% accuracy via co-adaptation, but a \emph{fresh attacker} recovers $>$85\% token accuracy on the same protected activations---the information was never removed, only hidden from one adversary. Stochastic perturbations $\boldsymbol{\delta} \perp \mathbf{h}$ can reduce MI (Remark~\ref{rem:stochastic}) but introduce inference-time non-determinism. Our bottleneck $G_i$, non-injective by construction ($\R^d \to \R^r$, $r \ll d$), achieves MI reduction deterministically.

\paragraph{Bottleneck dimension controls leakage.}
The bottleneck dimension $r$ provides an explicit, differentiable knob on information capacity. By the data processing inequality applied to the chain $x_j \to \mathbf{h}^{(k)}_j \to \mathbf{z}^{(k)}_j \to A_j$, where $\mathbf{z}^{(k)}_j \in \R^r$ is the bottleneck representation,
\begin{equation}
I(A_j; x_j) \leq H(\mathbf{z}^{(k)}_j) \leq r \cdot \log(2B/\epsilon),
\label{eq:bottleneck_bound}
\end{equation}
for activations bounded in $[-B,B]^r$ at precision $\epsilon$. Combined with Theorem~\ref{thm:privacy_bound}, this gives $\Prob[\cA(A){=}x] \leq (n r \log(2B/\epsilon){+}1)/\log|\cX^n|$, formalizing that smaller $r$ yields stronger privacy. The full statement (Theorem~\ref{thm:bottleneck_mi}) and corollary appear in Appendix~\ref{app:proofs}.

\paragraph{Utility--privacy tradeoff and selective protection.}
Two further consequences follow from the same information-theoretic framework and are formalized in the appendix. First, the chain $x \to A \to y$ with $y = x_{n+1}$ implies $I(A; y) \leq I(A; x)$ (Theorem~\ref{thm:tradeoff}, Appendix~\ref{app:proof-tradeoff-thm}): utility cannot exceed privacy leakage, so $\lambda$ and $r$ together navigate a hard frontier rather than an arbitrary tuning surface. Second, applying the token-level bound separately to sensitive positions $\cS$ and common positions $\cC$ yields:

\begin{corollary}[Selective Privacy Protection]
\label{cor:sensitive}
$\E[\Acc_{\text{sensitive}}] \leq \frac{1}{|\cS|}\sum_{i \in \cS}\frac{I(A; x_i)+1}{\log|\cV|}$. If the bottleneck disproportionately reduces $I(A; x_i)$ for high-entropy (rare, sensitive) tokens relative to low-entropy (common) tokens, then $\E[\Acc_{\text{sensitive}}] \ll \E[\Acc_{\text{common}}]$.
\end{corollary}

The intuition: low-entropy tokens (``the,'' ``is'') occupy a small region of activation space and are cheaply representable through the bottleneck; high-entropy tokens (unique words such as ``Gangrene,'' ``SSRI'', etc.) require many more bits to distinguish. When $r$ is limited, the adapter preferentially preserves common tokens (which also support syntactic coherence) and discards rare sensitive ones---the empirical pattern we report in Section~\ref{sec:experiments}.

%------------------------------------------------------------------
% IMPLEMENTATION
%------------------------------------------------------------------
% sections/implementation.tex
% Implementation details (compressed for 9-page limit)

\section{Implementation}
\label{sec:implementation}

We implement the framework in PyTorch~2.0 with HuggingFace Transformers (v4.35), using FP16 for the frozen backbone and FP32 for adapters and the adversary. We evaluate three LLMs: Mistral-7B-v0.1~\citep{jiang2023mistral}, LLaMA-2-7B~\citep{touvron2023llama}, and LLaMA-2-13B~\citep{touvron2023llama}. Code will be released upon publication. Full hyperparameters, the training algorithm, and complexity analysis appear in Appendix~\ref{app:impl}.

\paragraph{Adapter and adversary.}
Each adapter $G_i$ follows the bottleneck design of Eq.~\eqref{eq:adapter_structure}, with an output LayerNorm to stabilize activation magnitudes for downstream frozen layers. We fix $k{=}4$ adapter layers and sweep $r \in \{256, 512, 1024\}$ and $\lambda \in \{0.1, 0.5, 0.9\}$ as the two privacy-utility knobs. The adversary $q_\Psi$ is a 4-layer MLP with hidden dimensions $[d, 2d, 2d, |\cV|]$ and dropout $0.1$, performing per-position token classification. For LLaMA-2-7B, $q_\Psi$ has $\sim$181M parameters---an order of magnitude larger than the adapters ($\sim$17M)---so the adversary is never the bottleneck and any residual token information would be detected.

\paragraph{Training.}
Each step alternates (i) $S{=}3$ adversary updates on stop-gradient activations to tighten the MI upper bound, and (ii) one adapter update on the combined privacy-utility loss
\begin{equation}
\cL_{\text{total}} = \lambda \cdot \big(-\text{CE}(q_\Psi(A), x)\big) + (1-\lambda)\cdot \cL_{\text{NTP}}(A),
\label{eq:total_loss}
\end{equation}
where $\cL_{\text{NTP}}(A) = \text{CE}(T_L(\cdots T_{k+1}(A))_{<n}, x_{2:n})$ is computed by forwarding $A$ through the \emph{entire} frozen cloud pipeline. This ensures adapters learn representations compatible with downstream processing rather than merely matching original activations in some metric space. Multiple adversary steps per adapter step ($S{=}3$) follow standard adversarial training practice~\citep{goodfellow2014generative,arjovsky2017wasserstein}. Algorithm~\ref{alg:training} (Appendix~\ref{app:algorithm}) details the procedure.

\paragraph{Direct MI estimation.}
% The privacy loss directly estimates the variational MI upper bound (Lemma~\ref{lem:vclub}): the optimal adversary satisfies $q_{\Psi^*}(x|A) = p(x|A)$, at which point the bound is tight, and the adapter then maximizes the cross-entropy to push the bound down. Symmetrically, $\cL_{\text{NTP}}$ instantiates the utility lower bound from Lemma~\ref{lem:utility_lb}. Compared to cosine-similarity surrogates~\citep{qu2025prompt}, this avoids the Gaussian-activation assumption ($I = -\tfrac{1}{2}\log(1-\rho^2)$), eliminates sign ambiguity in the privacy gradient, and adapts continuously to the current adapter state.

The privacy loss tightens a variational bound on the per-position MI $\sum_j \MI{A_j}{x_j}$, which by Theorem~\ref{thm:token_acc} controls the success of position-wise attackers~\citep{luo2025prompt,qu2025prompt}; the optimal per-position adversary satisfies $q_{\Psi^*}(x_j|A_j) = p(x_j|A_j)$, at which the bound is tight, and the adapter then maximizes cross-entropy to drive it down (see Remark~\ref{rem:surrogate}).

%------------------------------------------------------------------
% EXPERIMENTS
%------------------------------------------------------------------
% sections/experiments.tex
% Experimental Evaluation (compressed for 9-page limit)

\section{Experimental Evaluation}
\label{sec:experiments}

\subsection{Experimental Setup}
\label{subsec:setup}

\paragraph{Datasets, models, and hardware.}
We evaluate on three datasets containing sensitive textual content: \textbf{Medical} (clinical notes from MedAlpaca with conditions, drugs, and symptoms), \textbf{Skytrax} (airline reviews with travel details and flight numbers), and \textbf{Legal} (ECHR case documents with names and personal circumstances). We use three pre-trained LLMs of increasing scale, as described above, all split at $k{=}4$ with frozen cloud layers. Training uses two RTX~4090 GPUs; latency is measured on a Jetson Orin Nano (8~GB) representing a realistic edge client.

\paragraph{Attacks and metrics.}
We evaluate against three prompt inversion attacks.
\textbf{\cite{luo2025prompt}} train a 4-layer MLP token classifier ($\sim$181M parameters) on intermediate activations under a black-box threat model where the adversary has no knowledge of the model components or any deployed defense.
\textbf{\cite{qu2025prompt}} use a constrained-optimization approach that searches for continuous embeddings matching the target activation, then discretizes to tokens.
\textbf{Ours} adopts the same 4-layer MLP architecture and cross-entropy training objective as Luo et al., but operates under the white-box threat model defined in Section~\ref{subsec:threat_model}: the adversary has full knowledge of the privacy adapter weights $\{G_1, \ldots, G_k\}$ and trains directly on activations sampled from the deployed defended model, following Kerckhoffs's principle that defense evaluation should not assume secrecy of its parameters. This represents the strongest classification-based adversary realizable in our setting.
Crucially, all three attacks are trained \emph{from scratch} on protected activations (\emph{fresh attacker}), not the co-adapted training adversary, so we evaluate against the strongest possible adversary for each defense configuration. We report (i) overall token accuracy ($\downarrow$), (ii) sensitive token accuracy ($\downarrow$), (iii) common token accuracy, and (iv) perplexity ($\downarrow$ for utility). Common tokens are the 118 most frequent stop words, punctuation, and structural words; sensitive tokens are the remainder. The \emph{privacy gap} = common $-$ sensitive accuracy quantifies selective protection.
\subsection{Main Results: Privacy-Utility Tradeoff and Bottleneck Ablation}
\label{subsec:main_results}

Table~\ref{tab:dim_comparison_all_attacks} reports attack accuracy and perplexity across models, datasets, bottleneck dimensions $r \in \{256, 512, 1024\}$, and tradeoff weights $\lambda \in \{0.1, 0.5, 0.9\}$, evaluated against three attacks.

Three patterns emerge. \emph{(i)} Increasing $\lambda$ consistently reduces attack accuracy at the cost of higher perplexity, instantiating the tradeoff predicted by Theorem~\ref{thm:tradeoff}; intermediate values ($\lambda{=}0.5$) achieve the best balance. \emph{(ii)} Larger $r$ increases information capacity, yielding both higher attack accuracy and lower perplexity, validating Theorem~\ref{thm:bottleneck_mi}'s prediction that bottleneck dimension directly bounds leakage. \emph{(iii)} The defense generalizes across attack types: fresh classification and optimization attackers track our training adversary closely, with no attack able to exploit a configuration where the others fail. This rules out the failure mode of Theorem~\ref{thm:residual} (a fresh attacker recovering information that the co-adapted adversary missed).

\begin{table*}[t]
\centering
\caption{Attack accuracy (\%) / Perplexity across attacks, models, datasets, and bottleneck dimensions ($r$, with $k{=}4$). Three attacks evaluated against the \emph{same} protected activations: (a) Ours (training adversary, co-adapted via minimax); (b)~\cite{luo2025prompt} (classification-based, fresh attacker); (c)~\cite{qu2025prompt} (optimization-based). Larger $r$ increases information capacity (Theorem~\ref{thm:bottleneck_mi}). Perplexity depends only on the defense (not the attack) and is therefore identical across (a), (b), (c) for each $(r, \lambda)$ configuration. \textbf{Bold} columns ($\lambda{=}0.5$) indicate the best privacy-utility operating point.}
\label{tab:dim_comparison_all_attacks}
\resizebox{\textwidth}{!}{%
\begin{tabular}{ll|c|ccc|ccc|ccc}
\toprule
& & & \multicolumn{3}{c|}{$r = 256$} & \multicolumn{3}{c|}{$r = 512$} & \multicolumn{3}{c}{$r = 1024$} \\
\cmidrule(lr){4-6} \cmidrule(lr){7-9} \cmidrule(lr){10-12}
\textbf{Model} & \textbf{Dataset} & \textbf{No Def.}
& $\lambda{=}0.1$ & $\boldsymbol{\lambda{=}0.5}$ & $\lambda{=}0.9$
& $\lambda{=}0.1$ & $\boldsymbol{\lambda{=}0.5}$ & $\lambda{=}0.9$
& $\lambda{=}0.1$ & $\boldsymbol{\lambda{=}0.5}$ & $\lambda{=}0.9$ \\
\midrule
\multicolumn{12}{c}{\textit{(a) Ours (training adversary, co-adapted via minimax)}} \\
\midrule
\multirow{3}{*}{Mistral-7B}
 & Skytrax & 89.0/44.7 & 86.4/45.6 & \textbf{53.3/52.1} & 42.1/83.0 & 88.1/44.1 & \textbf{58.1/52.5} & 50.0/79.8 & 88.8/44.0 & \textbf{60.0/52.0} & 56.0/70.1 \\
 & Medical & 85.8/29.4 & 84.6/29.4 & \textbf{60.5/30.1} & 36.7/45.2 & 85.1/29.4 & \textbf{65.5/29.9} & 36.7/45.0 & 85.8/29.0 & \textbf{50.1/29.0} & 50.0/44.4 \\
 & Legal   & 95.4/11.4 & 95.0/12.2 & \textbf{48.6/24.3} & 21.7/24.6 & 95.0/12.0 & \textbf{48.6/24.1} & 22.7/23.4 & 95.0/11.8 & \textbf{48.6/23.6} & 23.9/23.0 \\
\midrule
\multirow{3}{*}{LLaMA-2-7B}
 & Skytrax & 91.1/33.7 & 90.0/38.0 & \textbf{50.8/90.1} & 36.1/95.9 & 90.1/36.8 & \textbf{55.1/86.7} & 43.3/93.2 & 91.9/33.7 & \textbf{56.4/84.3} & 49.7/90.6 \\
 & Medical & 87.3/25.9 & 81.2/34.1 & \textbf{51.0/60.1} & 12.5/82.6 & 81.3/29.4 & \textbf{53.8/54.7} & 26.5/80.6 & 83.2/28.1 & \textbf{57.9/53.3} & 29.5/78.9 \\
 & Legal   & 95.5/7.8  & 88.4/9.4  & \textbf{51.2/11.9} & 32.0/30.9 & 91.0/9.4  & \textbf{54.7/10.7} & 33.1/26.1 & 95.5/9.4  & \textbf{58.4/10.0} & 33.7/25.9 \\
\midrule
\multirow{3}{*}{LLaMA-2-13B}
 & Skytrax & 95.4/16.9 & 89.0/20.1 & \textbf{39.1/34.0} & 24.1/140.0 & 92.1/17.7 & \textbf{49.3/32.1} & 31.8/132.0 & 92.3/17.9 & \textbf{50.0/30.9} & 32.0/130.0 \\
 & Medical & 96.2/25.1 & 96.0/30.0 & \textbf{38.9/40.0} & 16.1/60.0  & 95.9/25.9 & \textbf{42.5/39.9} & 22.8/60.7  & 96.1/25.2 & \textbf{42.0/39.1} & 22.0/59.9 \\
 & Legal   & 95.5/9.4  & 94.9/19.0 & \textbf{39.5/24.0} & 15.6/30.0  & 95.5/17.9 & \textbf{44.0/22.9} & 19.2/30.0  & 96.0/18.1 & \textbf{45.7/22.0} & 24.1/29.5 \\
\midrule
\multicolumn{12}{c}{\textit{(b)~\cite{luo2025prompt} (classification-based, fresh attacker)}} \\
\midrule
\multirow{3}{*}{Mistral-7B}
 & Skytrax & 90.2/44.7 & 88.1/45.6 & \textbf{56.5/52.1} & 35.4/83.0 & 89.6/44.1 & \textbf{60.0/52.5} & 42.8/79.8 & 90.0/44.0 & \textbf{60.2/52.0} & 48.6/70.1 \\
 & Medical & 87.1/29.4 & 86.0/29.4 & \textbf{60.8/30.1} & 39.5/45.2 & 86.7/29.4 & \textbf{67.8/29.9} & 39.5/45.0 & 87.1/29.0 & \textbf{52.6/29.0} & 52.4/44.4 \\
 & Legal   & 96.5/11.4 & 89.1/12.2 & \textbf{51.4/24.3} & 24.0/24.6 & 90.2/12.0 & \textbf{51.4/24.1} & 24.0/23.4 & 89.1/11.8 & \textbf{51.4/23.6} & 24.0/23.0 \\
\midrule
\multirow{3}{*}{LLaMA-2-7B}
 & Skytrax & 92.4/33.7 & 91.5/38.0 & \textbf{63.4/90.1} & 38.9/95.9 & 90.0/36.8 & \textbf{62.5/86.7} & 35.9/93.2 & 90.0/33.7 & \textbf{58.7/84.3} & 32.4/90.6 \\
 & Medical & 88.7/25.9 & 83.5/34.1 & \textbf{53.7/60.1} & 34.8/82.6 & 83.6/29.4 & \textbf{54.9/54.7} & 44.8/80.6 & 84.4/28.1 & \textbf{55.6/53.3} & 46.9/78.9 \\
 & Legal   & 96.7/7.8  & 92.9/9.4  & \textbf{53.8/11.9} & 34.5/30.9 & 94.7/9.4  & \textbf{54.0/10.7} & 42.6/26.1 & 96.0/9.4  & \textbf{55.7/10.0} & 43.0/25.9 \\
\midrule
\multirow{3}{*}{LLaMA-2-13B}
 & Skytrax & 96.5/16.9 & 90.4/20.1 & \textbf{41.6/34.0} & 26.5/140.0 & 93.5/17.7 & \textbf{51.7/32.1} & 34.1/132.0 & 93.7/17.9 & \textbf{52.5/30.9} & 34.4/130.0 \\
 & Medical & 97.3/25.1 & 93.0/30.0 & \textbf{41.4/40.0} & 18.4/60.0  & 95.2/25.9 & \textbf{45.0/39.9} & 25.2/60.7  & 97.1/25.2 & \textbf{45.4/39.1} & 24.3/59.9 \\
 & Legal   & 96.7/9.4  & 92.9/19.0 & \textbf{42.3/24.0} & 17.9/30.0  & 93.7/17.9 & \textbf{46.5/22.9} & 21.7/30.0  & 94.4/18.1 & \textbf{48.2/22.0} & 26.6/29.5 \\
\midrule
\multicolumn{12}{c}{\textit{(c)~\cite{qu2025prompt} (optimization-based)}} \\
\midrule
\multirow{3}{*}{Mistral-7B}
 & Skytrax & 88.4/44.7 & 84.7/45.6 & \textbf{49.8/52.1} & 33.2/83.0 & 85.3/44.1 & \textbf{54.5/52.5} & 40.6/79.8 & 85.9/44.0 & \textbf{56.4/52.0} & 45.0/70.1 \\
 & Medical & 84.6/29.4 & 82.1/29.4 & \textbf{56.4/30.1} & 32.5/45.2 & 82.7/29.4 & \textbf{61.0/29.9} & 33.8/45.0 & 84.0/29.0 & \textbf{47.6/29.0} & 47.4/44.4 \\
 & Legal   & 94.1/11.4 & 90.5/12.2 & \textbf{43.5/24.3} & 17.8/24.6 & 91.6/12.0 & \textbf{45.0/24.1} & 20.4/23.4 & 92.4/11.8 & \textbf{45.8/23.6} & 19.5/23.0 \\
\midrule
\multirow{3}{*}{LLaMA-2-7B}
 & Skytrax & 89.5/33.7 & 86.6/38.0 & \textbf{47.4/90.1} & 33.0/95.9 & 87.5/36.8 & \textbf{51.6/86.7} & 39.8/93.2 & 88.0/33.7 & \textbf{53.2/84.3} & 41.0/90.6 \\
 & Medical & 85.8/25.9 & 77.4/34.1 & \textbf{47.5/60.1} & 11.5/82.6 & 77.9/29.4 & \textbf{49.1/54.7} & 23.4/80.6 & 78.5/28.1 & \textbf{53.0/53.3} & 26.4/78.9 \\
 & Legal   & 94.0/7.8  & 84.5/9.4  & \textbf{46.0/11.9} & 26.6/30.9 & 87.1/9.4  & \textbf{49.4/10.7} & 27.2/26.1 & 94.0/9.4  & \textbf{53.0/10.0} & 27.5/25.9 \\
\midrule
\multirow{3}{*}{LLaMA-2-13B}
 & Skytrax & 94.2/16.9 & 80.8/20.1 & \textbf{34.7/34.0} & 19.6/140.0 & 86.0/17.7 & \textbf{44.5/32.1} & 27.4/132.0 & 88.2/17.9 & \textbf{45.4/30.9} & 27.6/130.0 \\
 & Medical & 95.1/25.1 & 83.6/30.0 & \textbf{36.2/40.0} & 12.0/60.0  & 85.9/25.9 & \textbf{39.6/39.9} & 16.0/60.7  & 88.7/25.2 & \textbf{40.3/39.1} & 17.9/59.9 \\
 & Legal   & 94.4/9.4  & 90.8/19.0 & \textbf{34.9/24.0} & 11.2/30.0  & 91.2/17.9 & \textbf{39.0/22.9} & 14.7/30.0  & 94.7/18.1 & \textbf{40.5/22.0} & 19.0/29.5 \\
\bottomrule
\end{tabular}
}
\end{table*}

\subsection{Sensitive Token Analysis}
\label{subsec:sensitive}

A central observation is that privacy protection is \emph{selective}: the bottleneck disproportionately suppresses sensitive tokens while preserving common structural ones, exactly as Corollary~\ref{cor:sensitive} predicts. Figure~\ref{fig:tradeoff_main} shows how sensitive token accuracy, NERR, and BLEU evolve with $\lambda$ across all three datasets; per-dataset breakdowns with the privacy-gap visualization appear in Appendix~\ref{app:per_dataset_ablation}.

\begin{figure*}[!htbp]
    \centering
    \begin{subfigure}[b]{0.32\textwidth}
        \centering
        \includegraphics[width=1\textwidth]{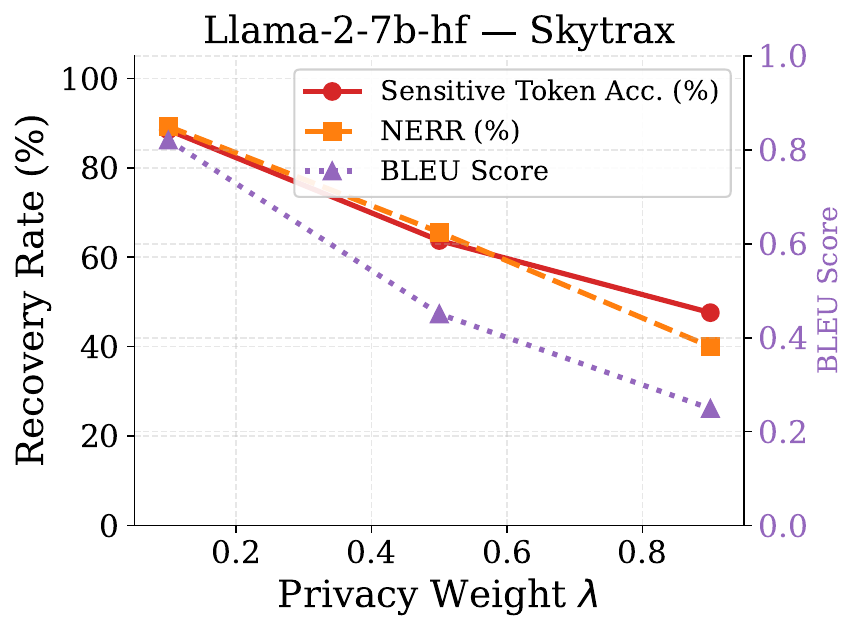}
        \caption{Skytrax}
        \label{fig:tradeoff_skytrax}
    \end{subfigure}
    \hfill
    \begin{subfigure}[b]{0.32\textwidth}
        \centering
        \includegraphics[width=1\textwidth]{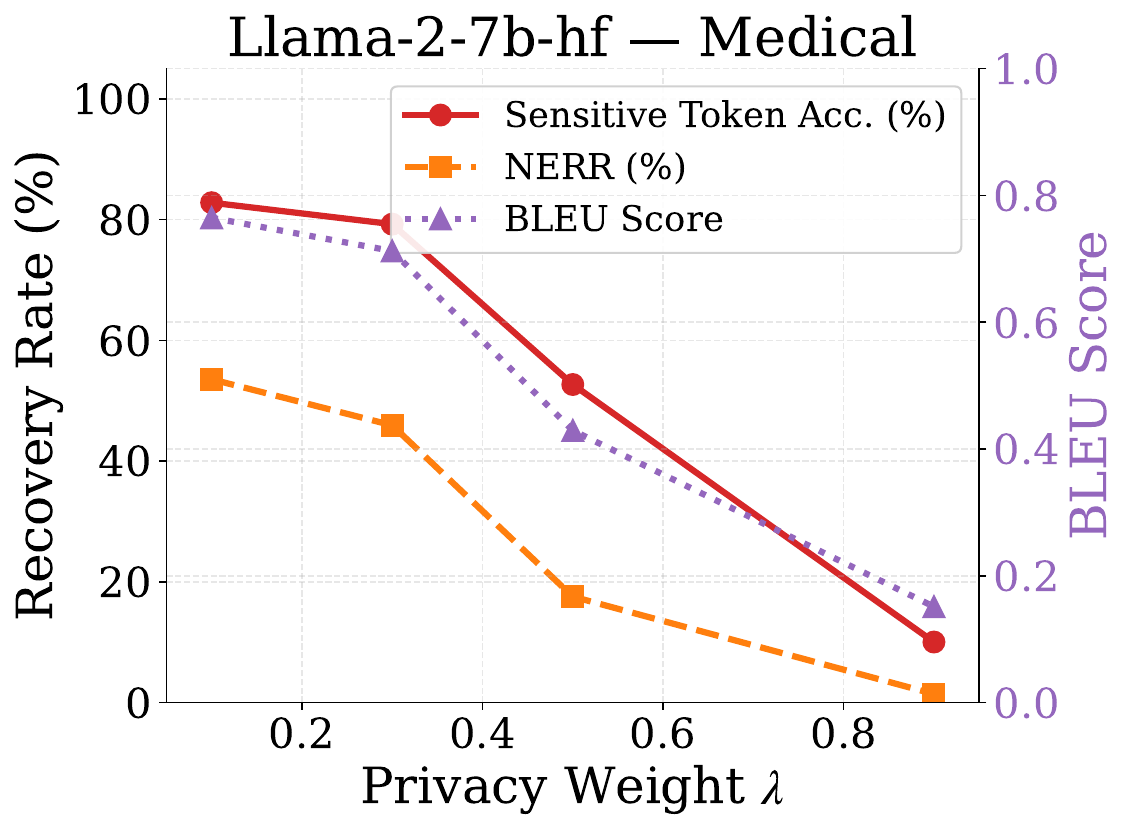}
        \caption{Medical}
        \label{fig:tradeoff_medical}
    \end{subfigure}
    \hfill
    \begin{subfigure}[b]{0.32\textwidth}
        \centering
        \includegraphics[width=1\textwidth]{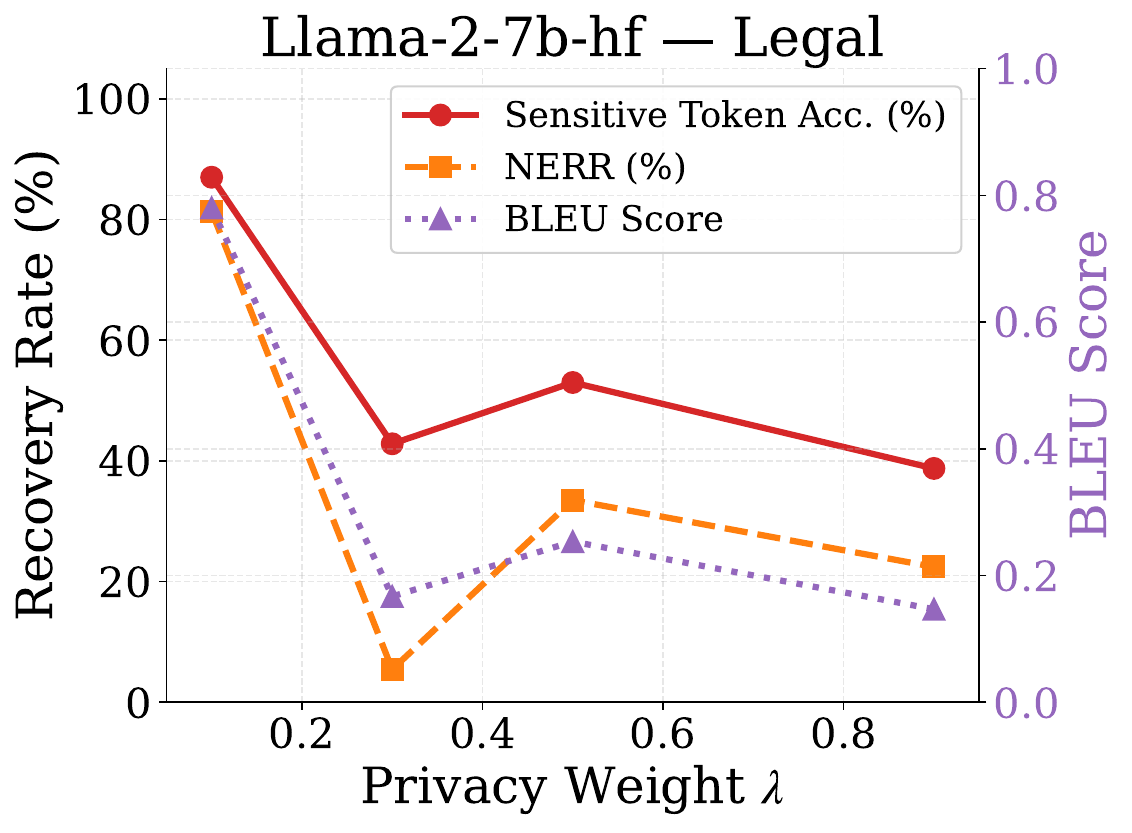}
        \caption{Legal}
        \label{fig:tradeoff_legal}
    \end{subfigure}
    \caption{Cross-dataset privacy-utility tradeoff over $\lambda$ (LLaMA-2-7B, $k{=}4$, $r{=}512$). As $\lambda$ increases, sensitive token accuracy and NERR decrease (stronger privacy) while BLEU also decreases (utility cost). Medical, with the most distinctive domain terminology, exhibits the strongest privacy gains.}
    \label{fig:tradeoff_main}
\end{figure*}

The privacy gap widens with $\lambda$: from 11--15\% at $\lambda{=}0.1$ to 30--35\% at $\lambda{=}0.9$. Common token accuracy stays high (88--99\%) even at aggressive settings, preserving syntactic coherence, while sensitive token accuracy drops from 85--89\% to 54--58\%---meaningful protection of domain-specific content. Notably, the defense provides this without explicit knowledge of which tokens are sensitive: the bottleneck discovers privacy-relevant compression purely from the minimax objective.

\subsection{Qualitative Analysis}
\label{subsec:qualitative}

Table~\ref{tab:qualitative_skytrax_moderate} shows a representative reconstruction on Skytrax at $\lambda{=}0.5$. The defense preserves syntactic structure and evaluative tokens (``good,'' ``friendly,'' ``service,'' ``staff'') while scrambling specific identifiers: city names (``Noumea'' $\to$ ``Lfrivebound''), route codes (``RUN-SEY'' $\to$ ``LAX-G and''), flight numbers (``621'' $\to$ ``58''), and dates (``4/2014'' $\to$ ``K/101ment''). The attacker can infer that an airline review was written but cannot recover \emph{which flight}, \emph{which route}, or \emph{when}---exactly the selective protection Corollary~\ref{cor:sensitive} predicts from bottleneck capacity arguments. Additional examples on Medical at $\lambda{=}0.3$ (where conditions and drug names are protected) and Skytrax at $\lambda{=}0.9$ (near-complete obfuscation) appear in Appendix~\ref{app:qualitative_examples}.

\begin{table}[t]
\centering
\caption{Reconstruction example (LLaMA-2-7B, Skytrax, $\lambda{=}0.5$, $r{=}256$). Sensitive tokens annotated as \textcolor{teal}{[PROT.]} or \textcolor{red}{[REC.]}. More examples in Appendix~\ref{app:qualitative_examples}.}
\label{tab:qualitative_skytrax_moderate}
\footnotesize
\setlength{\tabcolsep}{3pt}
\begin{tabular}{p{0.95\columnwidth}}
\toprule
\textbf{Example} (overall: 59.4\%) \\
\midrule
\textbf{Original:} Only flew Air Austral for medium haul (RUN-SEY) but was excellent! Good service much space good food and very friendly staff! \\
\textbf{Recovered:} only Flew Air Austral with1 holul (LAX-G and) but was full. nice service much space good food and very friendly staff. \\
\textbf{Sensitive:} ``medium'' $\to$ ``1'' \textcolor{teal}{[PROT.]}, ``RUN-SEY'' $\to$ ``LAX-G and'' \textcolor{teal}{[PROT.]}, ``excellent'' $\to$ ``full'' \textcolor{teal}{[PROT.]}, ``Good'' $\to$ ``nice'' \textcolor{teal}{[PROT.]}, ``Austral'' $\to$ ``Austral'' \textcolor{red}{[REC.]} \\
\bottomrule
\end{tabular}
\end{table}

\subsection{Comparison with Existing Defenses}
\label{sec:defense_comparison}

We compare our privacy adapters against four representative defense families operating under our constraint of frozen pre-trained layers and latency: \emph{noise perturbation} (Gaussian and Laplacian) injected at every device-side layer; \emph{LDP-style training} via NoPeek~\citep{vepakomma2020nopeek}, which adds a distance-correlation penalty between activations and inputs; \emph{information suppression} via Fisher-based dimension masking~\citep{guo2022bounding} and ShredMI random projection; and \emph{linear projection} via PCA at $k \in \{256, 512\}$. All baselines use the same fresh classification attacker on LLaMA-2-7B at split layer $k{=}4$. Full results are in Table~\ref{tab:baseline_comparison}.  We summarize the key findings here.

\begin{wraptable}{r}{.55\columnwidth}
\centering
\footnotesize
\caption{Defense comparison against the fresh classification attacker~\cite{luo2025prompt} (LLaMA-2-7B, $k{=}4$). Attack accuracy (\%, $\downarrow$) and perplexity ($\downarrow$) across three datasets. \textbf{Ours} achieves the lowest attack accuracy among all methods with usable perplexity (PPL $<$1000) on all three domains.}
\label{tab:baseline_comparison}
\setlength{\tabcolsep}{5pt}
\resizebox{.55\textwidth}{!}{
\begin{tabular}{l|cc|cc|cc}
\toprule
& \multicolumn{2}{c|}{\textbf{Skytrax}} & \multicolumn{2}{c|}{\textbf{Medical}} & \multicolumn{2}{c}{\textbf{Legal}} \\
\cmidrule(lr){2-3} \cmidrule(lr){4-5} \cmidrule(lr){6-7}
\textbf{Defense} & Atk & PPL & Atk & PPL & Atk & PPL \\
\midrule
No Defense                                  & 88.4 & 32.9    & 85.1 & 17.3    & 94.7 & 10.0 \\
\midrule
Gaussian $\sigma{=}0.1$                     & 76.4 & 46.0    & 68.7 & 21.4    & 89.4 & 12.3 \\
Gaussian $\sigma{=}0.5$                     & 17.2 & 2537    & 26.7 & 4222    & 56.2 & 903  \\
Laplace $b{=}0.1$                           & 64.9 & 75.2    & 61.1 & 31.5    & 85.8 & 17.4 \\
Laplace $b{=}0.5$                           & 10.9 & 7580    & 19.9 & 16316   & 43.9 & 2636 \\
\midrule
NoPeek                                      & 81.7 & 258.6   & 75.5 & 35.5    & 92.5 & 67.4 \\
Fisher                                      & 87.9 & 3766    & 84.4 & 10134   & 94.7 & 290.6 \\
ShredMI ($r{=}256$)                         & \phantom{0}8.2 & 4091 & 12.6 & 5213 & 17.9 & 1452 \\
PCA ($k{=}256$)                             & 87.2 & 46.5    & 82.6 & 26.5    & 94.6 & 12.1 \\
PCA ($k{=}512$)                             & 88.0 & 35.7    & 84.2 & 20.5    & 94.6 & 10.7 \\
\midrule
\textbf{Ours} ($r{=}256$, $\lambda{=}0.5$)  & \textbf{50.8} & \textbf{90.1} & \textbf{51.0} & \textbf{60.1} & \textbf{51.2} & \textbf{11.9} \\
\bottomrule
\end{tabular}
}
\end{wraptable}

All baselines exhibit failure modes our adapters avoid: noise faces a privacy-utility cliff (small budgets leave attack accuracy at 64.9--89.4\%, while $\sigma{=}0.5$ explodes perplexity to 903--4222); PCA preserves too much structure (within 1--2 points of no defense, confirming Theorem~\ref{thm:residual}); and MI-targeted methods (NoPeek~\citep{vepakomma2020nopeek}, Fisher~\citep{guo2022bounding}, ShredMI) trade utility for privacy at perplexities of 290--10{,}134. In contrast, our adapters at $r{=}256$, $\lambda{=}0.5$ achieve the lowest attack accuracy among defenses with usable perplexity on every dataset---50.8\% / 51.0\% / 51.2\% on Skytrax/Medical/Legal at PPL 90.1 / 60.1 / 11.9, a 37--43 point reduction versus no defense within an order of magnitude of baseline perplexity. Related families share the same limitations: LDP frameworks like SnD~\citep{mai2023split} exhibit the same cliff; cryptographic frameworks~\citep{juvekar2018gazelle,pang2024bolt,xu2025breaking} introduce orders-of-magnitude latency overhead (14 minutes per token on Llama-7B), incompatible with edge deployment.
% \subsection{Inference Overhead and Layer Persistence}
% \label{sec:overhead_layers}

% \paragraph{Inference latency.}
% Table~\ref{tab:latency} reports end-to-end latency on a Jetson Orin Nano client with an RTX~4090 server (64-token input, $k{=}4$). Privacy adapters add 5.6--8.4\% overhead across models---practical for edge deployment. The overhead does not grow with model size since it is dominated by the lightweight client-side adapter computation. Both adapter parameters and runtime scale linearly with $k$; we fix $k{=}4$ as a balance between privacy and client cost.

% \begin{table}[t]
% \centering
% \caption{Inference latency on Jetson Orin Nano (client) + RTX 4090 (cloud), $k{=}4$, 64-token input.}
% \label{tab:latency}
% \footnotesize
% \setlength{\tabcolsep}{4pt}
% \resizebox{.45\columnwidth}{!}{
% \begin{tabular}{lccc}
% \toprule
% \textbf{Model} & \textbf{Defense} & \textbf{Latency (ms)} & \textbf{Overhead} \\
% \midrule
% \multirow{2}{*}{Mistral-7B}
%  & None      & 54.3 & -- \\
%  & Adapters  & 58.8 & +8.25\% \\
% \midrule
% \multirow{2}{*}{LLaMA-2-7B}
%  & None      & 49.9 & -- \\
%  & Adapters  & 54.1 & +8.42\% \\
% \midrule
% \multirow{2}{*}{LLaMA-2-13B}
%  & None      & 77.0 & -- \\
%  & Adapters  & 81.4 & +5.65\% \\
% \bottomrule
% \end{tabular}
% }
% \end{table}

\paragraph{Inference latency.}
Privacy adapters add only 5.6--8.4\% end-to-end latency overhead across all three models on a Jetson Orin Nano client with RTX~4090 server (64-token input, $k{=}4$); full per-model breakdown in Appendix~\ref{app:latency} (Table~\ref{tab:latency}). The overhead does not grow with model size since it is dominated by the lightweight client-side adapter computation.%, and both adapter parameters and runtime scale linearly with $k$.

\paragraph{Token identity persists across all layers.}
We confirm the finding of~\cite{dong2025depth} that splitting at deeper layers does not by itself provide privacy: on LLaMA-2-7B / Medical without defense, fresh-attacker accuracy stays at 84.7\% / 80.7\% / 80.2\% / 81.1\% for split layers $k=4, 8, 12, 16$ respectively. Token identity persists through all 32 layers, motivating the need for an active defense rather than depth-based obfuscation.

%------------------------------------------------------------------
% CONCLUSIONS
%------------------------------------------------------------------
\section{Conclusions}
We introduced a principled, information-theoretic defense against prompt
inversion attacks in collaborative LLM inference. By deploying lightweight
privacy adapters that selectively suppress privacy-sensitive information,
our approach achieves strong privacy protection with minimal utility and
latency overhead. These results demonstrate that learned, task-aware
transformations are essential for practical and effective prompt privacy
in edge--cloud LLM systems.

%------------------------------------------------------------------
% BROADER IMPACTS (kept inline; NeurIPS does not require a fixed heading)
%------------------------------------------------------------------

%------------------------------------------------------------------
% ACKNOWLEDGMENTS (camera-ready only -- the `ack` env auto-hides at submission)
%------------------------------------------------------------------
% Uncomment for camera-ready. The `ack` environment is provided by the
% NeurIPS style file and is automatically hidden in the anonymized version.
%
% \begin{ack}
% We thank ... . This work was supported by ... .
% \end{ack}

%------------------------------------------------------------------
% REFERENCES
%------------------------------------------------------------------
% NeurIPS uses natbib (loaded by default by neurips_2026.sty).
% Any consistent style is acceptable; plainnat is the safe default.
\bibliographystyle{plainnat}
\bibliography{references,mobicom-ref}

%%%%%%%%%%%%%%%%%%%%%%%%%%%%%%%%%%%%%%%%%%%%%%%%%%%%%%%%%%%%%%%%%%%%%%%%%%%%
% NeurIPS PAPER CHECKLIST (REQUIRED)
% Lives in checklist.tex per the official NeurIPS 2026 template convention.
% You MUST fill this in -- skipping or removing it can be grounds for desk reject.
%%%%%%%%%%%%%%%%%%%%%%%%%%%%%%%%%%%%%%%%%%%%%%%%%%%%%%%%%%%%%%%%%%%%%%%%%%%%
\newpage

\section*{NeurIPS Paper Checklist}

\begin{enumerate}

\item {\bf Claims}
    \item[] Question: Do the main claims made in the abstract and introduction accurately reflect the paper's contributions and scope?
    \item[] Answer: \answerYes{}
    \item[] Justification: The abstract and introduction (Section~\ref{sec:introduction}) state our four main contributions: (i) an information-theoretic minimax framework for prompt inversion defense, (ii) privacy adapters as information-bottleneck modules with theoretical guarantees, (iii) a sensitive token analysis methodology, and (iv) experiments showing 15--35\% reduction in sensitive token recovery with $<$9\% latency overhead. These claims are supported by the formal results in Section~\ref{sec:theory} and the empirical evaluation in Section~\ref{sec:experiments}.

\item {\bf Limitations}
    \item[] Question: Does the paper discuss the limitations of the work performed by the authors?
    \item[] Answer: \answerYes{}
    \item[] Justification: We discuss several limitations throughout the paper. Theorem~\ref{thm:tradeoff} formalizes the inherent utility-privacy tradeoff: stronger privacy necessarily reduces achievable utility. Section~\ref{sec:experiments} reports the perplexity cost of stronger privacy settings ($\lambda{=}0.9$ degrades perplexity substantially). Our threat model assumes an honest-but-curious cloud server with white-box access; we do not address malicious servers that tamper with computation. Experiments are limited to three LLM families (LLaMA-2-7B/13B, Mistral-7B) at sequence length 32 with $k{=}4$ adapter layers; generalization to longer sequences, larger models, and other architectures requires further study. Appendix~\ref{app:noise_baselines} discusses why we do not compare against retrained representation-learning baselines, acknowledging this restricts our comparison set.

\item {\bf Theory assumptions and proofs}
    \item[] Question: For each theoretical result, does the paper provide the full set of assumptions and a complete (and correct) proof?
    \item[] Answer: \answerYes{}
    \item[] Justification: All theoretical results are numbered and cross-referenced. Theorem~\ref{thm:privacy_bound} (Privacy Leakage Bound), Theorem~\ref{thm:residual} (Residual MI Invariance), Theorem~\ref{thm:bottleneck_mi} (Bottleneck MI Bound), Theorem~\ref{thm:tradeoff} (Utility-Privacy Tradeoff), and Theorem~\ref{thm:token_acc} (Token Reconstruction Accuracy) are stated in Section~\ref{sec:theory} with informal intuition, with complete proofs and assumptions in Appendix~\ref{app:proofs}. Variational bounds (Lemmas~\ref{lem:vclub} and~\ref{lem:utility_lb}) used by the framework are proven in Appendices~\ref{app:vb} and~\ref{app:lem}. Assumptions (uniform prior on prompt space, bounded activations with precision $\epsilon$, deterministic invertible perturbations for the residual case) are stated explicitly in each theorem statement.

    \item {\bf Experimental result reproducibility}
    \item[] Question: Does the paper fully disclose all the information needed to reproduce the main experimental results of the paper to the extent that it affects the main claims and/or conclusions of the paper (regardless of whether the code and data are provided or not)?
    \item[] Answer: \answerYes{}
    \item[] Justification: Section~\ref{sec:implementation} and Appendix~\ref{app:impl} provide the full experimental setup. The privacy adapter architecture is given by Eq.~\eqref{eq:adapter_structure}; training procedure by Algorithm~\ref{alg:training} (Appendix~\ref{app:algorithm}); all hyperparameters in Table~\ref{tab:hyperparams} (Appendix~\ref{app:hyperparams}). Datasets (MedAlpaca, Skytrax airline reviews, ECHR Legal), models (LLaMA-2-7B/13B, Mistral-7B from HuggingFace), attack baselines (\cite{luo2025prompt} and~\cite{qu2025prompt}), and the optimizer/learning rate/training epochs are documented. Code will be released publicly upon publication.

\item {\bf Open access to data and code}
    \item[] Question: Does the paper provide open access to the data and code, with sufficient instructions to faithfully reproduce the main experimental results, as described in supplemental material?
    \item[] Answer: \answerNo{}
    \item[] Justification: Code is not released at submission time to preserve anonymity. The implementation will be released publicly upon publication, including training scripts, attack baselines, dataset preprocessing pipelines, and trained adapter checkpoints. All datasets used (MedAlpaca, Skytrax airline reviews, ECHR Legal) are publicly available; pre-trained models (LLaMA-2-7B/13B, Mistral-7B) are accessible via HuggingFace. The architecture (Eq.~\eqref{eq:adapter_structure}), training procedure (Algorithm~\ref{alg:training}), and hyperparameters (Table~\ref{tab:hyperparams}) are sufficient to reproduce our results from scratch.

\item {\bf Experimental setting/details}
    \item[] Question: Does the paper specify all the training and test details (e.g., data splits, hyperparameters, how they were chosen, type of optimizer) necessary to understand the results?
    \item[] Answer: \answerYes{}
    \item[] Justification: Section~\ref{sec:implementation} specifies optimizer (AdamW for adapters, Adam for adversary), learning rates ($10^{-4}$ adapter, $10^{-3}$ adversary), training duration (10 epochs), and the alternating minimax schedule ($S{=}3$ adversary steps per adapter step). Hyperparameter sweeps over $r \in \{256, 512, 1024\}$ and $\lambda \in \{0.1, 0.3, 0.5, 0.9\}$ are explicitly varied as primary controls. Section~\ref{subsec:setup} details datasets, train/test splits, and metrics. Full hyperparameter table appears in Appendix~\ref{app:hyperparams} (Table~\ref{tab:hyperparams}).

\item {\bf Experiment statistical significance}
    \item[] Question: Does the paper report error bars suitably and correctly defined or other appropriate information about the statistical significance of the experiments?
    \item[] Answer: \answerNo{}
    \item[] Justification: Main results in Table~\ref{tab:dim_comparison_all_attacks} and Figure~\ref{fig:tradeoff_main} report point estimates without error bars due to the computational cost of running multiple full training runs across all combinations of three models, three datasets, three bottleneck dimensions, and four privacy weights. The noise baseline tables in Appendix~\ref{app:noise_baselines} (Tables~\ref{tab:dp_skytrax}--\ref{tab:dp_legal}) do report mean $\pm$ standard deviation across test prompts, providing one indicator of variability. The trends across the multi-axis sweep ($\lambda$, $r$, model, dataset, attack type) are highly consistent, providing confidence in the reported observations beyond what single-cell error bars would convey.

\item {\bf Experiments compute resources}
    \item[] Question: For each experiment, does the paper provide sufficient information on the computer resources (type of compute workers, memory, time of execution) needed to reproduce the experiments?
    \item[] Answer: \answerYes{}
    \item[] Justification: Section~\ref{subsec:setup} states that training is performed on a server with two NVIDIA RTX 4090 GPUs (24~GB each), and latency benchmarking uses a Jetson Orin Nano (8~GB Ampere-class GPU) representing an edge client. Appendix~\ref{app:train} provides per-iteration computational complexity (Table~\ref{tab:complexity}) showing adapter overhead is approximately $r/d \approx 0.8\%$ of transformer cost. Adapter parameter counts per model are reported in Table~\ref{tab:param_count}. Inference latency on the edge client is in Table~\ref{tab:latency} (5.6--8.4\% overhead). The full hyperparameter sweep across 3 models $\times$ 3 datasets $\times$ 3 bottleneck dimensions $\times$ 4 privacy weights required additional preliminary runs not reported.
    
\item {\bf Code of ethics}
    \item[] Question: Does the research conducted in the paper conform, in every respect, with the NeurIPS Code of Ethics \url{https://neurips.cc/public/EthicsGuidelines}?
    \item[] Answer: \answerYes{}
    \item[] Justification: This research conforms to the NeurIPS Code of Ethics. The work is defense-oriented, aimed at protecting user prompt privacy in collaborative LLM inference. We use only publicly available datasets and pre-trained models with appropriate citations. No human subjects research, crowdsourcing, or sensitive personal data collection is involved; the datasets used (MedAlpaca, Skytrax, ECHR) consist of pre-existing public corpora. Anonymity is preserved in the submission.

\item {\bf Broader impacts}
    \item[] Question: Does the paper discuss both potential positive societal impacts and negative societal impacts of the work performed?
    \item[] Answer: \answerYes{}
    \item[] Justification: We include a Broader Impacts section discussing both positive and negative societal impacts. Positive: enabling privacy-preserving LLM inference for sensitive applications such as healthcare, legal services, and enterprise retrieval-augmented generation, lowering the barrier to deploying privacy-aware AI on resource-constrained devices. Negative: as with any privacy-enhancing technology, improved confidentiality could be misused to obscure malicious activities; we note our defense is designed for benign user data under defined threat models and does not prevent lawful access at higher system layers, and we encourage complementary mechanisms for transparency and accountability.
    
\item {\bf Safeguards}
    \item[] Answer: \answerNA{}
    \item[] Justification: The paper proposes a defense mechanism (privacy adapters) rather than releasing pre-trained generative models, attack tools, or scraped datasets that pose high misuse risk. The trained adapter modules are tied to specific frozen pre-trained LLMs and are designed to protect privacy rather than generate or extract content. Released code will consist of training scripts and lightweight adapter checkpoints; standard responsible-release practices apply.

\item {\bf Licenses for existing assets}
    \item[] Question: Are the creators or original owners of assets (e.g., code, data, models), used in the paper, properly credited and are the license and terms of use explicitly mentioned and properly respected?
    \item[] Answer: \answerYes{}
    \item[] Justification: All existing assets are properly credited. Pre-trained models: LLaMA-2-7B/13B~\cite{touvron2023llama} (Meta LLaMA-2 Community License), Mistral-7B-v0.1~\cite{jiang2023mistral} (Apache 2.0). Datasets: MedAlpaca (clinical query corpus, CC BY 4.0), Skytrax airline reviews (publicly available), and ECHR case law (public domain). Software: PyTorch 2.0 (BSD-style), HuggingFace Transformers v4.35 (Apache 2.0). Baseline attack implementations follow~\cite{luo2025prompt} and~\cite{qu2025prompt}, both cited. All licenses are respected in our experimental usage.

\item {\bf New assets}
    \item[] Question: Are new assets introduced in the paper well documented and is the documentation provided alongside the assets?
    \item[] Answer: \answerYes{}
    \item[] Justification: The new assets introduced are the privacy adapter modules and the training framework. Documentation is provided in Section~\ref{sec:method} (architecture), Section~\ref{sec:implementation} (training), and Appendix~\ref{app:impl} (hyperparameters, complexity analysis, full algorithm). At submission time, code and adapter checkpoints are not released to preserve anonymity; upon publication, all assets will be released with documentation, license, and usage instructions.

\item {\bf Crowdsourcing and research with human subjects}
    \item[] Question: For crowdsourcing experiments and research with human subjects, does the paper include the full text of instructions given to participants and screenshots, if applicable, as well as details about compensation (if any)? 
    \item[] Answer: \answerNA{}
    \item[] Justification: This work does not involve crowdsourcing or research with human subjects. All datasets used are pre-existing public corpora.

\item {\bf Institutional review board (IRB) approvals or equivalent for research with human subjects}
    \item[] Question: Does the paper describe potential risks incurred by study participants, whether such risks were disclosed to the subjects, and whether Institutional Review Board (IRB) approvals (or an equivalent approval/review based on the requirements of your country or institution) were obtained?
    \item[] Answer: \answerNA{}
    \item[] Justification: This work does not involve human subjects research and therefore does not require IRB approval. All experimental data comes from pre-existing public corpora.

\item {\bf Declaration of LLM usage}
    \item[] Question: Does the paper describe the usage of LLMs if it is an important, original, or non-standard component of the core methods in this research? Note that if the LLM is used only for writing, editing, or formatting purposes and does \emph{not} impact the core methodology, scientific rigor, or originality of the research, declaration is not required.
    \item[] Answer: \answerNA{}
    \item[] Justification: LLMs (LLaMA-2-7B/13B, Mistral-7B) are used as the experimental targets being defended, not as a methodological component. They are pre-trained, frozen, and serve as the underlying model whose intermediate activations our privacy adapters protect. No LLMs were used as core, original, or non-standard components in the development of our defense methodology, theoretical analysis, or evaluation pipeline.

\end{enumerate}

%%%%%%%%%%%%%%%%%%%%%%%%%%%%%%%%%%%%%%%%%%%%%%%%%%%%%%%%%%%%%%%%%%%%%%%%%%%%
% APPENDIX
% NeurIPS is single-column throughout, so no \onecolumn needed.
%%%%%%%%%%%%%%%%%%%%%%%%%%%%%%%%%%%%%%%%%%%%%%%%%%%%%%%%%%%%%%%%%%%%%%%%%%%%
\appendix

\section{Technical Appendices and Supplementary Material}
% Note from the official template: think of the appendix as ``optional reading''
% for reviewers. The paper must stand alone without it; do not put critical
% experiments here.

\appendix
\section{Notation}
\label{app:notation}

Table~\ref{tab:notation} summarizes the notation used throughout the paper.

\begin{table}[h]
\centering
\footnotesize
\caption{Notation summary.}
\label{tab:notation}
\setlength{\tabcolsep}{3pt}
\resizebox{.85\columnwidth}{!}{
\begin{tabular}{@{}ll ll@{}}
\toprule
\textbf{Symbol} & \textbf{Description}
& \textbf{Symbol} & \textbf{Description} \\
\midrule
$\cX$ & Token space
& $\text{LM\_head}$ & Language modeling head \\

$x \in \cX^n$ & Input prompt (length $n$)
& $\mathcal{V}$ & Vocabulary ($|\mathcal{V}|$ tokens) \\

$\Emb$ & Embedding function
& $y = x_{n+1}$ & Next-token target \\

$T_i$ & $i$-th transformer layer
& $\MI{X}{Y}$ & Mutual information \\

$G_i$ & $i$-th privacy adapter
& $\Ent{X}$ & Entropy \\

$A$ & Protected intermediate activation
& $r$ & Bottleneck dimension \\

$d$ & Hidden dimension
& $\lambda$ & Privacy-utility tradeoff weight \\

$q_\Psi$ & Adversary (token classifier)
& $k$ & Split layer index \\
\bottomrule
\end{tabular}
}
\end{table}

\section{Proofs}
\label{app:proofs}

This appendix contains full statements and proofs of all theoretical results referenced in Section~\ref{sec:theory}, together with auxiliary results (variational bounds, definitions, and surrogate-loss propositions).

\subsection{Definitions}
\label{app:definitions}

\begin{definition}[Reconstruction Error]
\label{def:reconstruction_error}
For a prompt $x$ and its reconstruction $\hat{x}$, we measure error as:
\begin{equation}
\text{Error}(x, \hat{x}) = \norm{x - \hat{x}}_p
\end{equation}
for some $p$-norm (e.g., $p=2$ for Euclidean distance in embedding space).
\end{definition}

\begin{definition}[Reconstruction Attack Set]
\label{def:attack_set}
Let $\cA_{\PIA}$ denote the set of all prompt inversion attacks:
\begin{equation}
\cA_{\PIA} = \left\{ \cA: \R^{n \times d} \to \cX^n \right\},
\end{equation}
where each attack $\cA$ maps activation $A$ to reconstructed prompt $\hat{x} = \cA(A)$.
\end{definition}

\begin{definition}[Token Accuracy]
\label{def:token_accuracy}
For prompts of length $n$, token accuracy is:
\begin{equation}
\Acc_{\text{token}} = \frac{1}{n} \sum_{i=1}^n \mathbb{1}[\hat{x}_i = x_i]
\end{equation}
\end{definition}

\subsection{Variational Upper Bound (vCLUB)}
\label{app:vb}

\begin{lemma}[Variational Upper Bound, vCLUB]
\label{lem:vclub}
For any distribution $q_\Psi(x|A)$,
\begin{equation}
\MI{A}{x} \leq \E_{p(A,x)}[\log q_\Psi(x|A)] - \E_{p(A)p(x)}[\log q_\Psi(x|A)].
\end{equation}
\end{lemma}

\begin{proof}
We decompose the mutual information as:
\begin{align*}
\MI{A}{x} &= \E_{p(A,x)}\left[\log \frac{p(x|A)}{p(x)}\right] \\
&= \E_{p(A,x)}[\log p(x|A)] - \E_{p(x)}[\log p(x)] \\
&= \E_{p(A,x)}\left[\log \frac{p(x|A)}{q_\Psi(x|A)}\right] + \E_{p(A,x)}[\log q_\Psi(x|A)] \\
&\quad - \E_{p(x)}[\log p(x)] \\
&= \E_{p(A)}[\KL{p(x|A)}{q_\Psi(x|A)}] \\
&\quad + \E_{p(A,x)}[\log q_\Psi(x|A)] - \E_{p(x)}[\log p(x)].
\end{align*}
Since $\KL{p(x|A)}{q_\Psi(x|A)} \geq 0$ and $\E_{p(x)}[\log p(x)]$ is constant w.r.t.\ the encoder:
\begin{equation*}
\MI{A}{x} \leq \E_{p(A,x)}[\log q_\Psi(x|A)] - \E_{p(A)p(x)}[\log q_\Psi(x|A)] + C.
\end{equation*}
\end{proof}

\subsection{Variational Lower Bound for Utility}
\label{app:lem}

\begin{lemma}[Variational Lower Bound]
\label{lem:utility_lb}
For the language model distribution $p_{\text{LM}}(y|A)$ where $y = x_{n+1}$ is the next token,
\begin{equation}
\MI{A}{y} \geq \E_{p(A,y)}[\log p_{\text{LM}}(y|A)] + \Ent{y}.
\end{equation}
\end{lemma}

\begin{proof}
\begin{align*}
\MI{A}{y} &= \Ent{y} - \Hcond{y}{A} \\
&= \Ent{y} + \E_{p(A,y)}[\log p(y|A)] \\
&= \Ent{y} + \E_{p(A,y)}[\log p_{\text{LM}}(y|A)] \\
&\quad + \E_{p(A,y)}\left[\log \frac{p(y|A)}{p_{\text{LM}}(y|A)}\right] \\
&= \Ent{y} + \E_{p(A,y)}[\log p_{\text{LM}}(y|A)] \\
&\quad + \E_{p(A)}[\KL{p(y|A)}{p_{\text{LM}}(y|A)}].
\end{align*}
Since $\KL{p(y|A)}{p_{\text{LM}}(y|A)} \geq 0$:
\begin{equation*}
\MI{A}{y} \geq \E_{p(A,y)}[\log p_{\text{LM}}(y|A)] + \Ent{y}.
\end{equation*}
\end{proof}

\subsection{Tightening the Variational Upper Bound}
\label{app:tightening}

To tighten the variational upper bound from Lemma~\ref{lem:vclub}, we parameterize $q_\Psi(x|A)$ as a per-position token classifier:
\begin{equation}
q_\Psi(x \mid A) = \prod_{j=1}^{n} q_\Psi(x_j \mid A_j),
\end{equation}
where each factor is implemented as an MLP that maps the activation at position $j$ to a distribution over the vocabulary:
\begin{equation}
q_\Psi(x_j \mid A_j) = \text{softmax}\left(f_\Psi(A_j)\right), \quad j = 1, \ldots, n,
\end{equation}
with $f_\Psi: \mathbb{R}^d \to \mathbb{R}^{|\mathcal{V}|}$ and $A_j \in \mathbb{R}^d$. This factorization is consistent with the attack model of~\cite{luo2025prompt}.

Tightening the bound requires minimizing the KL divergence between the true posterior and the variational approximation:
\begin{align}
&\min_\Psi \KL{p(x|A)}{q_\Psi(x|A)} \nonumber \\
&= \min_\Psi \E_{p(A,x)}\left[\log p(x|A) - \log q_\Psi(x|A)\right] \nonumber \\
&\Leftrightarrow \max_\Psi \E_{p(A,x)}\left[\log q_\Psi(x|A)\right],
\end{align}
where $\E[\log p(x|A)]$ is constant in $\Psi$ and dropped. Substituting the factorized form:
\begin{equation}
\Leftrightarrow \min_\Psi \underbrace{-\E_{p(A,x)}\left[\sum_{j=1}^{n} \log q_\Psi(x_j \mid A_j)\right]}_{\text{Per-position cross-entropy}}.
\end{equation}
% Thus, training the adversary to minimize per-position cross-entropy is equivalent to tightening the variational MI upper bound.

Thus, training the adversary to minimize per-position cross-entropy is equivalent to maximizing the per-position variational lower bound on $\sum_j I(A_j; x_j)$ (see Remark~\ref{rem:surrogate} for why this is a principled surrogate rather than a direct bound on the joint $I(A;x)$).

% \begin{remark}
% The per-position factorization yields a \emph{lower bound} on the true variational objective $\max_\Psi \E[\log q_\Psi(x|A)]$, since it ignores cross-position dependencies. The MI upper bound is therefore \emph{looser} than what a sequence-level adversary could achieve. Our defense is thus evaluated conservatively: any reduction in the per-position bound necessarily reduces the true MI as well.
% \end{remark}

\begin{remark}[Per-position factorization as a principled surrogate]
\label{rem:surrogate}
The per-position factorization $q_\Psi(x|A) = \prod_{j=1}^n q_\Psi(x_j|A_j)$ is structurally unable to model cross-position dependencies in $p(x|A)$, since transformer activations $A_j$ encode contextual information about positions $i\neq j$ via attention prior to the bottleneck. Consequently, the variational quantity tightened by training $\Psi$ does not directly upper-bound the joint $\MI{A}{x}$; rather, it tightens a bound on the per-position MI $\sum_j \MI{A_j}{x_j}$.

We adopt this surrogate deliberately. \textbf{(i) Threat-model alignment.} State-of-the-art prompt inversion attacks operate per-position: \citet{luo2025prompt} is a per-position token classifier and \citet{qu2025prompt} performs per-token optimization. By Theorem~\ref{thm:token_acc}, the expected token-level recovery of any such attack is bounded as $\E[\Acc_{\text{token}}] \le \tfrac{1}{n}\sum_j (\MI{A_j}{x_j}+1)/\log|\cV|$, so reducing $\sum_j \MI{A_j}{x_j}$ directly tightens the success bound for the relevant adversary class. \textbf{(ii) Connection to joint MI.} By Lemma~\ref{lem:per_position_bound}, $\sum_j \MI{A}{x_j} \ge \MI{A}{x}$, so the marginal sum upper-bounds the joint MI. Although our adversary estimates $\MI{A_j}{x_j} \le \MI{A}{x_j}$ by the data processing inequality, the architectural argument in Remark~\ref{rem:architectural} shows that residual joint leakage beyond what is captured per-position must be inherited from the frozen backbone rather than created by the defense. Closing the gap via a sequence-level autoregressive adversary $q_\Psi(x_j|A,x_{<j})$ is left to future work.
\end{remark}

\subsection{Per-Position MI Upper Bound}
\label{app:per_position}

\begin{lemma}[Per-Position MI Upper Bound]
\label{lem:per_position_bound}
Let $A=(A_1,\ldots,A_n)\in\R^{n\times d}$ be the protected activation and $x=(x_1,\ldots,x_n)\in\cV^n$ the input prompt. Then
\begin{equation}
\MI{A}{x} \;\le\; \sum_{j=1}^{n} \MI{A}{x_j}.
\end{equation}
\end{lemma}

\begin{proof}
By the chain rule of mutual information,
\begin{equation}
\MI{A}{x} = \MI{A}{x_1,\ldots,x_n} = \sum_{j=1}^{n} \MI{A}{x_j \mid x_{<j}}.
\end{equation}
For each term, since conditioning reduces entropy,
\begin{align*}
\MI{A}{x_j \mid x_{<j}}
&= \Hcond{x_j}{x_{<j}} - \Hcond{x_j}{A,x_{<j}} \\
&\le \Ent{x_j} - \Hcond{x_j}{A} = \MI{A}{x_j},
\end{align*}
where the inequality combines $\Hcond{x_j}{x_{<j}} \le \Ent{x_j}$ and $\Hcond{x_j}{A,x_{<j}} \ge \Hcond{x_j}{A}$ (both follow because conditioning on more variables cannot increase entropy). Summing over $j$ yields the claim.
\end{proof}

\begin{remark}[Architectural compatibility of the surrogate]
\label{rem:architectural}
Our adapters $G_i$ are pointwise modules that operate on each position independently; they cannot introduce new cross-position redistribution of token information. Cross-position information flow is inherited from the (frozen, unmodified) attention layers $T_i$ and is therefore the same across all defense configurations. Consequently, any residual joint leakage that would not be captured by $\sum_j \MI{A_j}{x_j}$ comes from the frozen backbone rather than from the defense, ensuring that reductions in our per-position objective reflect genuine information destruction by the bottleneck.
\end{remark}

\subsection{Proof of Theorem~\ref{thm:privacy_bound}}
\label{app:proof-tradeoff}

\begin{proof}
We apply Fano's inequality~\citep{cover2006elements} to bound the reconstruction probability.

\textbf{Step 1: Fano's Inequality.}
For any estimator $\hat{x} = \cA(A)$:
\begin{equation}
\Hcond{x}{A} \leq h_b(P_e) + P_e \log(\abs{\cX^n} - 1),
\end{equation}
where $P_e = \Prob[\hat{x} \neq x]$ and $h_b(\cdot)$ is binary entropy.

\textbf{Step 2: Relate Conditional Entropy to MI.}
By definition, $\Hcond{x}{A} = \Ent{x} - \MI{x}{A}$.

\textbf{Step 3: Combine.}
Using $h_b(P_e) \leq 1$ and $\log(\abs{\cX^n}-1) \leq \log\abs{\cX^n}$:
\begin{equation}
\Ent{x} - \MI{x}{A} \leq 1 + P_e \log\abs{\cX^n}.
\end{equation}

\textbf{Step 4: Solve.}
Under a uniform prior, $\Ent{x} = \log\abs{\cX^n}$. Rearranging:
\begin{equation}
1 - P_e \leq \frac{\MI{x}{A} + 1}{\log\abs{\cX^n}}.
\end{equation}
Since $1 - P_e = \Prob[\cA(A) = x]$, the bound follows.
\end{proof}

\begin{corollary}
\label{cor:privacy}
For fixed prompt-space size $|\cX^n|$, minimizing $\MI{x}{A}$ directly reduces the upper bound on adversary success. As $\MI{x}{A} \to 0$, success probability approaches $\frac{1}{\log|\cX^n|}$, negligible for large prompt spaces.
\end{corollary}

\subsection{Bottleneck Mutual Information Bound}
\label{app:proof-bottleneck}

We now formalize the bottleneck argument sketched in Eq.~\eqref{eq:bottleneck_bound} of the main text.

\begin{theorem}[Bottleneck MI Bound]
\label{thm:bottleneck_mi}
Let $G_i(\mathbf{h}) = W_{\text{up}}^{(i)} \cdot \sigma(W_{\text{down}}^{(i)} \cdot \mathbf{h})$ be a privacy adapter with bottleneck dimension $r$. Let $\mathbf{z}^{(i)} = \sigma(W_{\text{down}}^{(i)} \cdot \mathbf{h}^{(i)}) \in \mathbb{R}^r$ denote the bottleneck representation at position $j$. Then by the data processing inequality:
\begin{equation}
I(A_j; x_j) \leq I(\mathbf{z}_j^{(k)}; x_j) \leq H(\mathbf{z}_j^{(k)}).
\end{equation}
For bounded activations with $\mathbf{z}_j^{(k)} \in [-B, B]^r$, the entropy is bounded as:
\begin{equation}
I(A_j; x_j) \leq r \cdot \log(2B / \epsilon),
\end{equation}
where $\epsilon$ is the precision of the representation. Thus, smaller bottleneck dimension $r$ directly constrains the information capacity of the transmitted activation.
\end{theorem}

\begin{proof}
The first inequality follows from the data processing inequality applied to the Markov chain $x_j \to \mathbf{h}^{(k)}_j \to \mathbf{z}_j^{(k)} \to A_j$, where $A_j = W_{\text{up}}^{(k)} \mathbf{z}_j^{(k)}$ is a deterministic function of $\mathbf{z}_j^{(k)}$. The second follows from the entropy bound $H(\mathbf{z}_j^{(k)}) \leq r \cdot \log(2B/\epsilon)$ for $r$-dimensional bounded continuous variables discretized at precision $\epsilon$.
\end{proof}

\begin{corollary}[Bottleneck Dimension Controls Privacy]
\label{cor:bottleneck}
Combining Theorems~\ref{thm:privacy_bound} and~\ref{thm:bottleneck_mi}, for any attack $\cA \in \cA_{\PIA}$:
\begin{equation}
\Prob[\cA(A) = x] \leq \frac{n \cdot r \cdot \log(2B/\epsilon) + 1}{\log |\cX^n|}.
\end{equation}
This bound decreases with smaller bottleneck dimension $r$, formalizing the intuition that narrower bottlenecks provide stronger privacy guarantees.
\end{corollary}

% \subsection{Proof of Theorem~\ref{thm:residual} (Residual MI Invariance)}
% \label{app:proof-residual}

% \begin{proof}
% The mapping $\phi: \mathbf{h} \mapsto \mathbf{h} + \boldsymbol{\delta}(\mathbf{h})$ is a deterministic function of $\mathbf{h}$. If $\boldsymbol{\delta}$ is continuous and has bounded gradient (as is the case for neural network adapters with bounded weights), then $\phi$ is invertible in a neighborhood of each point. For an invertible transformation $\phi$, the data processing inequality holds with equality: $I(\phi(\mathbf{h}); x) = I(\mathbf{h}; x)$. Even when $\phi$ is not globally invertible, we have $I(\phi(\mathbf{h}); x) \leq I(\mathbf{h}; x)$ with equality when $\phi$ is injective. Since $\phi(\mathbf{h}) = \mathbf{h} + \boldsymbol{\delta}(\mathbf{h})$ contains $\mathbf{h}$ as a component, any adversary can approximate $\mathbf{h}$ from $\phi(\mathbf{h})$, yielding $I(\phi(\mathbf{h}); x) \geq I(\mathbf{h}; x)$. Combined: $I(\mathbf{h} + \boldsymbol{\delta}(\mathbf{h}); x) = I(\mathbf{h}; x)$.
% \end{proof}

% \begin{remark}
% For \emph{stochastic} perturbations $\boldsymbol{\delta} \sim p(\boldsymbol{\delta})$ independent of $\mathbf{h}$, the data processing inequality gives $I(\mathbf{h} + \boldsymbol{\delta}; x) \leq I(\mathbf{h}; x)$, and MI can decrease. However, stochastic noise at inference time introduces non-deterministic model behavior, making outputs unreproducible---a significant drawback for deployed systems. Our bottleneck approach achieves MI reduction deterministically.
% \end{remark}

\subsection{Proof of Theorem~\ref{thm:residual} (Residual MI Invariance)}
\label{app:proof-residual}

\begin{proof}
Suppose $\phi: \mathbf{h} \mapsto \mathbf{h} + \boldsymbol{\delta}(\mathbf{h})$ is injective. Then there exists a measurable left-inverse $\phi^{-1}: \phi(\R^d) \to \R^d$ satisfying $\phi^{-1}(\phi(\mathbf{h})) = \mathbf{h}$ for all $\mathbf{h}$.

Applying the data processing inequality to the Markov chain $x \to \mathbf{h} \to \phi(\mathbf{h})$:
\begin{equation}
I(\phi(\mathbf{h}); x) \leq I(\mathbf{h}; x).
\end{equation}
Conversely, applying the data processing inequality to the chain $x \to \phi(\mathbf{h}) \to \phi^{-1}(\phi(\mathbf{h})) = \mathbf{h}$:
\begin{equation}
I(\mathbf{h}; x) = I(\phi^{-1}(\phi(\mathbf{h})); x) \leq I(\phi(\mathbf{h}); x).
\end{equation}
Combining the two inequalities yields $I(\phi(\mathbf{h}); x) = I(\mathbf{h}; x)$.

\textbf{Sufficient condition for injectivity.} If $\boldsymbol{\delta}$ is $L$-Lipschitz with $L < 1$, then for any $\mathbf{h}_1 \neq \mathbf{h}_2$, the reverse triangle inequality gives
\begin{align}
\|\phi(\mathbf{h}_1) - \phi(\mathbf{h}_2)\| 
&\geq \|\mathbf{h}_1 - \mathbf{h}_2\| - \|\boldsymbol{\delta}(\mathbf{h}_1) - \boldsymbol{\delta}(\mathbf{h}_2)\| \nonumber \\
&\geq (1 - L)\|\mathbf{h}_1 - \mathbf{h}_2\| > 0,
\end{align}
so $\phi$ is injective (in fact, bi-Lipschitz with constants $1-L$ and $1+L$).
\end{proof}

\begin{remark}[Scope and counterexamples]
\label{rem:residual_scope}
Theorem~\ref{thm:residual} requires injectivity and does not extend to all deterministic $\boldsymbol{\delta}$. Two illustrative counterexamples: (i) $\boldsymbol{\delta}(\mathbf{h}) = -\mathbf{h}$ yields $\phi(\mathbf{h}) \equiv \mathbf{0}$, giving $I(\phi(\mathbf{h}); x) = 0$; (ii) any linear residual $\phi(\mathbf{h}) = (I + B)\mathbf{h}$ with $I + B$ singular collapses information along the null space of $I + B$. Our claim is restricted to the operating regime of residual defenses proposed in the literature: LoRA-style adapters with standard small initialization, learned additive perturbations with bounded weights, and bounded-norm noise mechanisms all satisfy $L < 1$ Lipschitz in practice and therefore fall under the theorem. The empirical finding (Section~\ref{sec:experiments}) that residual adapters yield $>85\%$ fresh-attacker recovery is consistent with this restricted regime.
\end{remark}

\begin{remark}[Stochastic perturbations]
\label{rem:stochastic}
For \emph{stochastic} perturbations $\boldsymbol{\delta} \sim p(\boldsymbol{\delta})$ independent of $\mathbf{h}$, the data processing inequality gives $I(\mathbf{h} + \boldsymbol{\delta}; x) \leq I(\mathbf{h}; x)$, and MI can strictly decrease. Stochastic noise is therefore not covered by Theorem~\ref{thm:residual}. However, stochastic noise at inference time introduces non-deterministic model behavior, making outputs unreproducible---a significant drawback for deployed systems. Our bottleneck $G_i$ achieves MI reduction deterministically by replacing rather than perturbing the activation, mapping $\R^d \to \R^r$ with $r \ll d$ (a non-injective compression).
\end{remark}

\subsection{Surrogate Loss Bounds}
\label{app:surrogate}

\begin{proposition}[Privacy Surrogate Bounds MI]
\label{prop:privacy_surrogate}
Let $A$ be the protected activation and $E(x)$ the prompt embedding. Then
\begin{equation}
I(A; x) \leq I(A; E(x))
\end{equation}
by the data processing inequality, since $x \to E(x) \to A$ forms a Markov chain. For approximately Gaussian activations, $I(A; E(x))$ is monotonically increasing in $|\text{CosSim}(A, E(x))|$, so minimizing cosine similarity reduces the upper bound on $I(A; x)$.
\end{proposition}

\begin{proposition}[Utility Surrogate Preserves Log-Likelihood]
\label{prop:utility_surrogate}
Let $A_{\text{clean}} = F_1(E(x))$. If $\text{CosSim}(A, A_{\text{clean}}) \geq 1 - \epsilon$, then
\begin{equation}
|\E[\log p_{\text{LM}}(y|F_2(A))] - \E[\log p_{\text{LM}}(y|F_2(A_{\text{clean}}))]| \leq L \cdot \|A - A_{\text{clean}}\|_2,
\end{equation}
where $L$ is the Lipschitz constant of $\log p_{\text{LM}} \circ F_2$. Maximizing cosine similarity therefore preserves the variational lower bound on $I(A; y)$.
\end{proposition}

\subsection{Utility--Privacy Tradeoff}
\label{app:proof-tradeoff-thm}

\begin{theorem}[Inherent Utility-Privacy Tradeoff]
\label{thm:tradeoff}
For the Markov chain $x \to A \to y$ where $A = (G_k \circ T_k \circ \cdots \circ G_1 \circ T_1)(\Emb(x))$ and $y = x_{n+1}$:
\begin{equation}
I(A; y) \leq I(A; x) \leq H(x).
\end{equation}
Thus, decreasing privacy leakage $I(A; x)$ constrains the maximum achievable utility $I(A; y)$.
\end{theorem}

\begin{proof}
\textbf{Step 1: Bound on Utility.}
By the data processing inequality on the Markov chain $x \to A \to y$:
\begin{equation}
I(A; y) \leq I(x; y) \leq H(y).
\end{equation}

\textbf{Step 2: Privacy-utility coupling.}
By the chain rule, $I(A; x, y) = I(A; x) + I(A; y \mid x)$. Since $y = x_{n+1}$ is deterministic given $x$, $I(A; y \mid x) = 0$, so
\begin{equation}
I(A; x, y) = I(A; x).
\end{equation}
Applying the chain rule in the other order yields
\begin{equation}
I(A; x) = I(A; y) + I(A; x \mid y).
\end{equation}
Since $I(A; x \mid y) \geq 0$:
\begin{equation}
I(A; y) \leq I(A; x) \leq H(x).
\end{equation}
Utility cannot exceed privacy leakage. The tradeoff parameter $\lambda$ in Eq.~\eqref{eq:combined_objective} controls this allocation.
\end{proof}

\begin{corollary}
\label{cor:tradeoff}
For a fixed prompt entropy $\Ent{x}$, increasing utility $\MI{A}{y}$ necessarily increases potential leakage $\MI{A}{x}$. The tradeoff parameter $\lambda$ and the bottleneck dimension $r$ together control how the available information budget is allocated between task utility and privacy protection.
\end{corollary}

\subsection{Token-Level Reconstruction Accuracy}
\label{app:proof-token-bound}

\begin{theorem}[Token Reconstruction Accuracy Bound]
\label{thm:token_acc}
For any prompt inversion attack $\mathcal{A} \in \mathcal{A}_{PIA}$, the expected token accuracy is bounded as:
\begin{equation}
\mathbb{E}[\text{Acc}_{\text{token}}] \leq \frac{1}{n} \sum_{i=1}^{n} \frac{I(A; x_i) + 1}{\log |\mathcal{V}|},
\end{equation}
where $|\mathcal{V}|$ is the vocabulary size.
\end{theorem}

\begin{proof}
For position $i$, the attacker produces $\hat{x}_i = \mathcal{A}(A)_i$ with accuracy $\text{Acc}_i = 1 - P_{e,i}$. Applying Fano's inequality per-token:
\begin{equation}
H(x_i \mid A) \leq h_b(P_{e,i}) + P_{e,i} \log_2(|V| - 1).
\end{equation}
Since $H(x_i \mid A) = H(x_i) - I(A; x_i)$ and $H(x_i) = \log_2|V|$ under uniform prior:
\begin{equation}
\log_2|V| - I(A; x_i) \leq 1 + P_{e,i} \log_2|V|.
\end{equation}
Rearranging:
\begin{equation}
\text{Acc}_i = 1 - P_{e,i} \leq \frac{I(A; x_i) + 1}{\log_2|V|}.
\end{equation}
Averaging over $n$ positions yields
\begin{equation}
\E[\text{Acc}_{\text{token}}] \leq \frac{1}{n} \sum_{i=1}^{n} \frac{I(A; x_i) + 1}{\log_2|V|}.
\end{equation}
\end{proof}

The sensitive-vs-common token corollary (Corollary~\ref{cor:sensitive}) is stated in the main text and follows immediately by restricting the average above to positions $i \in \cS$.

\section{Implementation Details}
\label{app:impl}

\subsection{System Constraints}
\label{app:constraints}

\begin{table}[h]
\centering
\footnotesize
\caption{System constraints on the client device.}
\label{tab:constraints}
\setlength{\tabcolsep}{4pt}
\begin{tabular}{@{}ll@{}}
\toprule
\textbf{Constraint} & \textbf{Description} \\
\midrule
Layer & $k \leq k_{\max}$ (max.\ transformer layers on device) \\
Parameter & $\sum_{i=1}^k |G_i| \leq B_{\text{params}}$ (adapter parameter budget) \\
Latency & $T_{\text{device}}(\{G_i\}) \leq T_{\max}$ (inference time overhead) \\
\bottomrule
\end{tabular}
\end{table}

\subsection{Hyperparameters}
\label{app:hyperparams}

Table~\ref{tab:hyperparams} summarizes all training hyperparameters used in our experiments. We fix $k{=}4$ adapter layers throughout (split after the 4th transformer layer), which provides a practical balance between on-device computation and privacy protection. The bottleneck dimension $r$ and privacy weight $\lambda$ are the primary knobs for controlling the privacy-utility tradeoff and are varied systematically in our experiments.

\begin{table}[h]
\centering
\footnotesize
\caption{Privacy adapter hyperparameters.}
\label{tab:hyperparams}
\setlength{\tabcolsep}{3pt}
\resizebox{.85\columnwidth}{!}{
\begin{tabular}{@{}ll ll@{}}
\toprule
\textbf{Hyperparameter} & \textbf{Value}
& \textbf{Hyperparameter} & \textbf{Value} \\
\midrule
Bottleneck dim.\ $r$ & $\{256, 512, 1024\}$
& Num.\ adapter layers $k$ & $4$ \\

Nonlinearity & GELU
& Initialization & Xavier uniform \\

Adapter LR $\eta_{\text{adapt}}$ & $10^{-4}$
& Adversary LR $\eta_{\text{adv}}$ & $10^{-3}$ \\

Adversary steps $S$ & $3$
& Gradient clipping & $1.0$ \\

Optimizer (adapters) & AdamW ($\beta_1{=}0.9$, wd${=}0.01$)
& Optimizer (adversary) & Adam ($\beta_1{=}0.9$) \\

Privacy weight $\lambda$ & $\{0.1, 0.3, 0.5, 0.9\}$
& Training epochs & $10$ \\

Batch size & $1$
& Max sequence length & $32$ \\
\bottomrule
\end{tabular}
}
\end{table}

\subsection{Training Algorithm}
Details of our algorithm are shown in Algorithm~\ref{alg:training}.
\label{app:algorithm}

\begin{algorithm}[h]
\caption{Privacy Adapter Training via Minimax Optimization}
\label{alg:training}
\begin{algorithmic}[1]
\REQUIRE Frozen LLM layers $\{T_1, \ldots, T_L\}$, split layer $k$, dataset $\mathcal{D}$, privacy weight $\lambda$, bottleneck dim.\ $r$, adversary steps $S$
\ENSURE Trained privacy adapters $\{G_1, \ldots, G_k\}$
\STATE Initialize adapters $G_i: \mathbb{R}^d \to \mathbb{R}^r \to \mathbb{R}^d$ (Xavier uniform)
\STATE Initialize adversary $q_\Psi$ as 4-layer MLP
\FOR{each epoch}
  \FOR{each batch $\mathbf{x} \sim \mathcal{D}$}
    \STATE \textit{// Forward through device with interleaved adapters}
    \STATE $\mathbf{h}^{(0)} \gets \text{Embed}(\mathbf{x})$
    \FOR{$i = 1, \ldots, k$}
      \STATE $\mathbf{h}^{(i)} \gets G_i(T_i(\mathbf{h}^{(i-1)}))$ \hfill $\triangleright$ No residual
    \ENDFOR
    \STATE $\mathbf{A} \gets \mathbf{h}^{(k)}$
    \STATE
    \STATE \textit{// Phase 1: Adversary update (tighten MI bound)}
    \FOR{$s = 1, \ldots, S$}
      \STATE $\mathcal{L}_{\text{adv}} \gets \text{CE}(q_\Psi(\texttt{stopgrad}(\mathbf{A})), \mathbf{x})$
      \STATE Update $\Psi$ to minimize $\mathcal{L}_{\text{adv}}$
    \ENDFOR
    \STATE
    \STATE \textit{// Phase 2: Adapter update (privacy-utility tradeoff)}
    \STATE Recompute $\mathbf{A}$ with gradient flow through $\{G_i\}$
    \STATE $\mathcal{L}_{\text{priv}} \gets -\text{CE}(q_\Psi(\mathbf{A}), \mathbf{x})$ \hfill $\triangleright$ Maximize CE
    \STATE $\mathbf{y} \gets T_L(\cdots T_{k+1}(\mathbf{A}))$
    \STATE $\mathcal{L}_{\text{util}} \gets \text{CE}(\mathbf{y}_{<n}, \mathbf{x}_{2:n})$
    \STATE $\mathcal{L}_{\text{total}} \gets \lambda \cdot \mathcal{L}_{\text{priv}} + (1 - \lambda) \cdot \mathcal{L}_{\text{util}}$
    \STATE Update $\{G_i\}$ to minimize $\mathcal{L}_{\text{total}}$
  \ENDFOR
\ENDFOR
\RETURN $\{G_1, \ldots, G_k\}$
\end{algorithmic}
\end{algorithm}

\subsection{Overhead Analysis}
\label{app:train}

\paragraph{Training complexity.}
Table~\ref{tab:complexity} reports complexity per training iteration, where $n$ is sequence length, $k$ device layers, $d$ hidden dim., $r$ adapter rank, $L$ total layers, and $N_{\text{iter}}$ adversary iterations.

\begin{table}[h]
\centering
\footnotesize
\caption{Computational complexity per training iteration.}
\label{tab:complexity}
\begin{tabular}{lcc}
\toprule
\textbf{Component} & \textbf{Forward} & \textbf{Backward} \\
\midrule
Frozen transformer layers & $O(nkd^2)$ & $O(0)$ \\
Privacy adapters & $O(nkdr)$ & $O(nkdr)$ \\
Adversary (PIA) & $O(N_{\text{iter}} \cdot nd)$ & $O(N_{\text{iter}} \cdot nd)$ \\
Cloud forward (frozen) & $O(n(L-k)d^2)$ & $O(0)$ \\
\bottomrule
\end{tabular}
\end{table}

The adapter overhead $O(nkdr)$ is substantially smaller than $O(nkd^2)$ since $r \ll d$—approximately $r/d \approx 0.8\%$ of transformer computation in practice. Gradients flow only through adapter parameters, reducing backward-pass memory by over 99\%.

\paragraph{Storage overhead.}
Adapter parameter count is $k \times 2 \times d \times r$. Table~\ref{tab:param_count} shows totals per model.

\begin{table}[h]
\centering
\small
\caption{Privacy adapter parameter count ($r{=}32$, $k{=}4$).}
\label{tab:param_count}
\setlength{\tabcolsep}{4pt}
\begin{tabular}{@{}lccc@{}}
\toprule
\textbf{Model} & $d$ & \textbf{Adapter params} & \textbf{\% of model} \\
\midrule
Mistral-7B   & 4096 & 1.05M & 0.015\% \\
LLaMA-2-7B   & 4096 & 1.05M & 0.015\% \\
LLaMA-2-13B  & 5120 & 1.31M & 0.010\% \\
\bottomrule
\end{tabular}
\end{table}

\section{Additional Experimental Results}
\label{app:additional_results}

\subsection{Per-Dataset Privacy-Utility Ablation}
\label{app:per_dataset_ablation}

This section provides per-dataset breakdowns of the privacy-utility ablation summarized in Section~\ref{subsec:sensitive}. For each dataset, we show three panels: (a) token accuracy with separate curves for overall, common, and sensitive tokens (the shaded \emph{privacy gap} is the region between common and sensitive recovery); (b) BLEU score; and (c) named entity recovery rate (NERR).

These views reveal the mechanism behind the cross-dataset tradeoff curves of Figure~\ref{fig:tradeoff_main}: as $\lambda$ increases, the sensitive-token curve drops while the common-token curve remains stable, producing the widening privacy gap predicted by Corollary~\ref{cor:sensitive}.

\begin{figure}[h]
    \centering
    \includegraphics[width=0.95\textwidth]{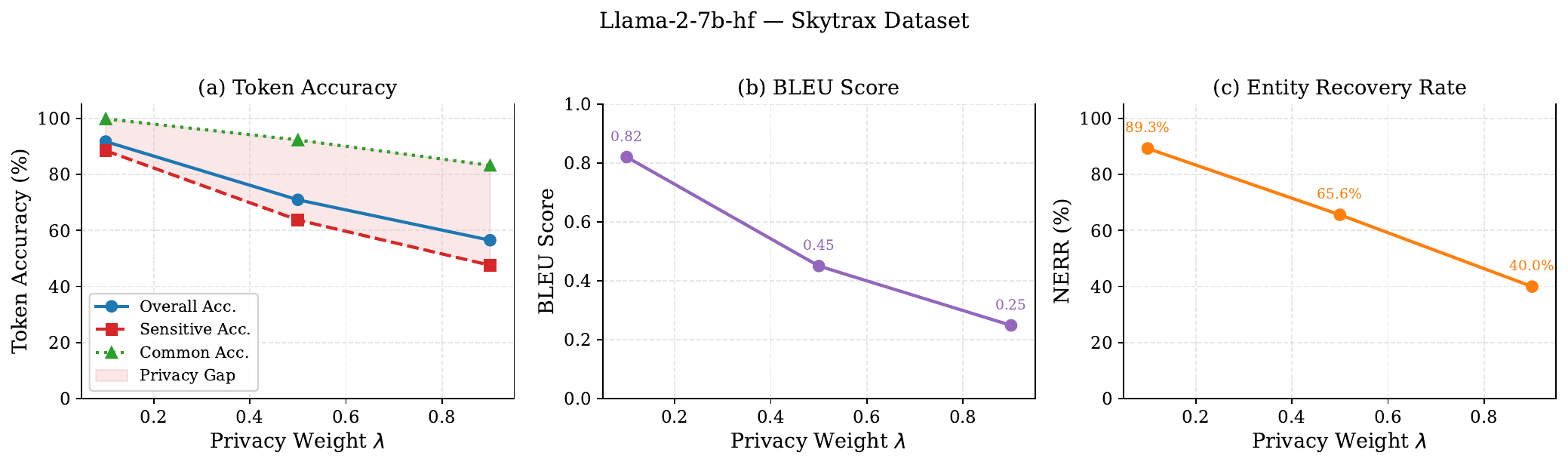}
    \caption{Skytrax dataset (LLaMA-2-7B, $k{=}4$, $r{=}512$). The privacy gap widens as $\lambda$ increases.}
    \label{fig:ablation_skytrax}
\end{figure}

\begin{figure}[h]
    \centering
    \includegraphics[width=0.95\textwidth]{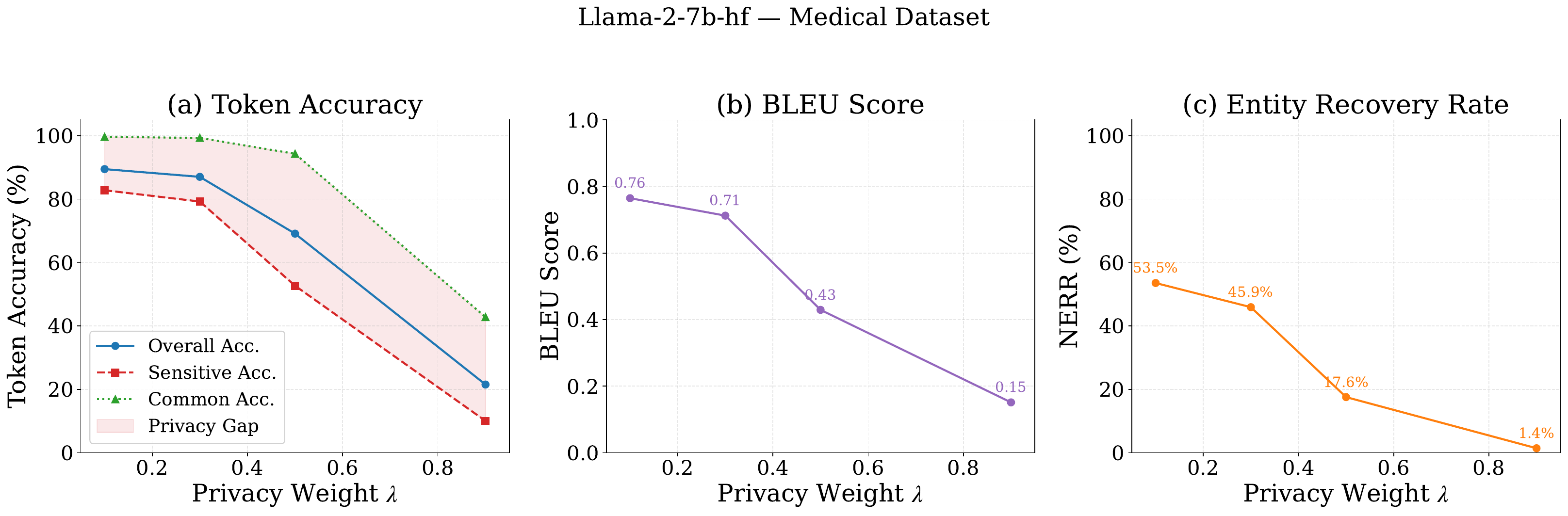}
    \caption{Medical dataset (LLaMA-2-7B, $k{=}4$, $r{=}512$). Sensitive token accuracy drops most aggressively, reflecting the high concentration of domain-specific terminology (medical conditions, drug names).}
    \label{fig:ablation_medical}
\end{figure}

\begin{figure}[h]
    \centering
    \includegraphics[width=0.95\textwidth]{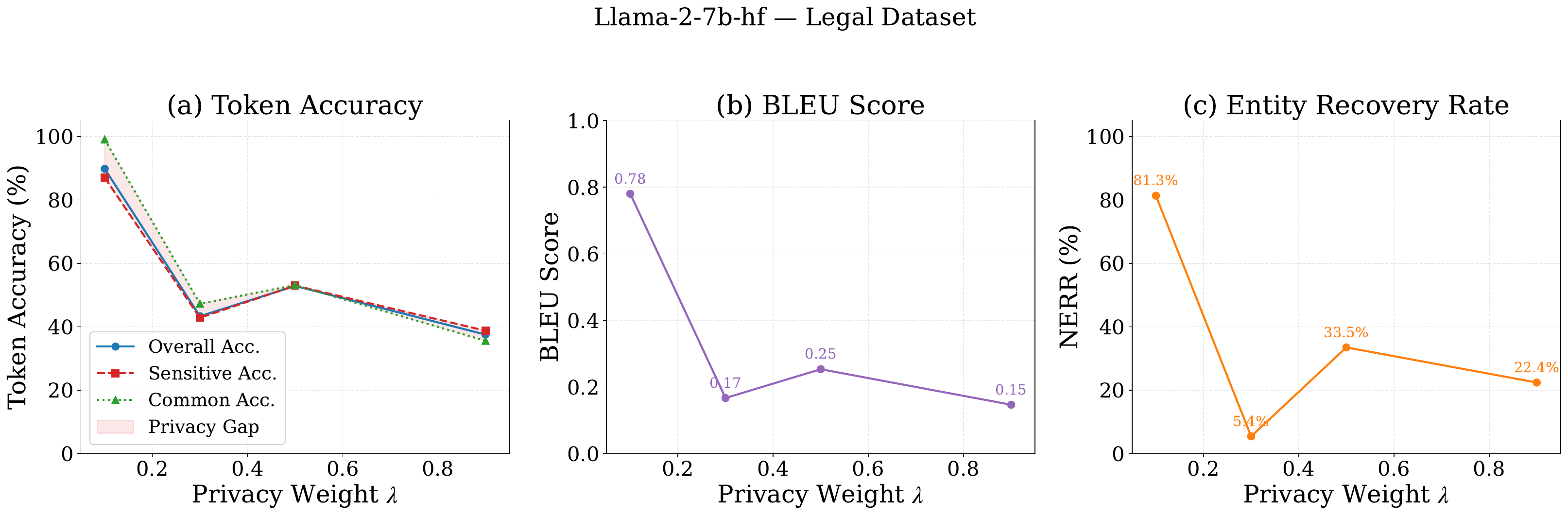}
    \caption{Legal dataset (LLaMA-2-7B, $k{=}4$, $r{=}512$). Token accuracy curves converge at intermediate $\lambda$, indicating more uniform entropy across token types in legal text.}
    \label{fig:ablation_legal}
\end{figure}

\subsection{Defense Baseline Comparison}
\label{app:baseline_comparison}

Table~\ref{tab:baseline_comparison} reports the full defense comparison referenced in Section~\ref{sec:defense_comparison}, evaluating noise (Gaussian, Laplacian), LDP-style training (NoPeek~\cite{vepakomma2020nopeek}), Fisher-based information suppression~\cite{guo2022bounding}, ShredMI random projection, and PCA against the same fresh classification attacker~\cite{luo2025prompt} on LLaMA-2-7B at split $k{=}4$. Our privacy adapters at $r{=}256$, $\lambda{=}0.5$ achieve the lowest attack accuracy among all defenses with usable perplexity on every dataset.

\begin{table}[h]
\centering
\footnotesize
\caption{Defense comparison against the fresh classification attacker~\cite{luo2025prompt} (LLaMA-2-7B, $k{=}4$). Attack accuracy (\%, $\downarrow$) and perplexity ($\downarrow$) across three datasets. \textbf{Ours} achieves the lowest attack accuracy among all methods with usable perplexity (PPL $<$1000) on all three domains.}
\label{tab:baseline_comparison}
\setlength{\tabcolsep}{5pt}
\begin{tabular}{l|cc|cc|cc}
\toprule
& \multicolumn{2}{c|}{\textbf{Skytrax}} & \multicolumn{2}{c|}{\textbf{Medical}} & \multicolumn{2}{c}{\textbf{Legal}} \\
\cmidrule(lr){2-3} \cmidrule(lr){4-5} \cmidrule(lr){6-7}
\textbf{Defense} & Atk & PPL & Atk & PPL & Atk & PPL \\
\midrule
No Defense                                  & 88.4 & 32.9    & 85.1 & 17.3    & 94.7 & 10.0 \\
\midrule
Gaussian $\sigma{=}0.1$                     & 76.4 & 46.0    & 68.7 & 21.4    & 89.4 & 12.3 \\
Gaussian $\sigma{=}0.5$                     & 17.2 & 2537    & 26.7 & 4222    & 56.2 & 903  \\
Laplace $b{=}0.1$                           & 64.9 & 75.2    & 61.1 & 31.5    & 85.8 & 17.4 \\
Laplace $b{=}0.5$                           & 10.9 & 7580    & 19.9 & 16316   & 43.9 & 2636 \\
\midrule
NoPeek                                      & 81.7 & 258.6   & 75.5 & 35.5    & 92.5 & 67.4 \\
Fisher                                      & 87.9 & 3766    & 84.4 & 10134   & 94.7 & 290.6 \\
ShredMI ($r{=}256$)                         & \phantom{0}8.2 & 4091 & 12.6 & 5213 & 17.9 & 1452 \\
PCA ($k{=}256$)                             & 87.2 & 46.5    & 82.6 & 26.5    & 94.6 & 12.1 \\
PCA ($k{=}512$)                             & 88.0 & 35.7    & 84.2 & 20.5    & 94.6 & 10.7 \\
\midrule
\textbf{Ours} ($r{=}256$, $\lambda{=}0.5$)  & \textbf{50.8} & \textbf{90.1} & \textbf{51.0} & \textbf{60.1} & \textbf{51.2} & \textbf{11.9} \\
\bottomrule
\end{tabular}
\end{table}

\subsection{Inference Latency Breakdown}
\label{app:latency}

Table~\ref{tab:latency} reports per-model end-to-end latency on a Jetson Orin Nano client (representing a realistic resource-constrained edge device) with an RTX~4090 cloud server, 64-token input, and split layer $k{=}4$. Privacy adapter overhead ranges from 5.6\% (LLaMA-2-13B) to 8.4\% (LLaMA-2-7B), well within practical limits for edge deployment. Notably, overhead does not grow with model size because the additional computation introduced by the adapters is fixed by the bottleneck dimension $r$ rather than by the underlying transformer width.

\begin{table}[h]
\centering
\caption{Inference latency on Jetson Orin Nano (client) + RTX 4090 (cloud), $k{=}4$, 64-token input.}
\label{tab:latency}
\footnotesize
\setlength{\tabcolsep}{4pt}
\begin{tabular}{lccc}
\toprule
\textbf{Model} & \textbf{Defense} & \textbf{Latency (ms)} & \textbf{Overhead} \\
\midrule
\multirow{2}{*}{Mistral-7B}
 & None      & 54.3 & -- \\
 & Adapters  & 58.8 & +8.25\% \\
\midrule
\multirow{2}{*}{LLaMA-2-7B}
 & None      & 49.9 & -- \\
 & Adapters  & 54.1 & +8.42\% \\
\midrule
\multirow{2}{*}{LLaMA-2-13B}
 & None      & 77.0 & -- \\
 & Adapters  & 81.4 & +5.65\% \\
\bottomrule
\end{tabular}
\end{table}

\subsection{Comparison with Noise-Based Defenses}
\label{app:noise_baselines}

We do not include comparisons to learned representation-learning baselines—such as adversarial representation learning~\citep{zhao2020trade}, information bottleneck variants~\citep{jaiswal2020invariant}, Inf2Guard-style MI frameworks~\citep{noorbakhsh2024inf2guard}, dimensionality reduction with retraining~\citep{dong2025depth}, or inversion-aware training—for three reasons:

\begin{enumerate}[itemsep=0pt,topsep=2pt]
\item Most existing learned defenses are designed for supervised classification and require training a full encoder jointly with the downstream model. Applying these to collaborative LLM inference would require retraining large portions of the transformer stack on the client side, which is infeasible. Our method operates under the realistic constraint of frozen pre-trained layers.
\item MI-based frameworks like Inf2Guard~\citep{noorbakhsh2024inf2guard} target dataset-level privacy (membership/property inference) using global MI estimates over training distributions. Extending these to token-level prompt inversion in autoregressive LLMs requires nontrivial redesign and is beyond our scope.
\item Dimensionality reduction or inversion-aware retraining incurs high computational and tuning costs, making fair edge-cloud comparison infeasible.
\end{enumerate}

We instead compare against noise-based defenses, which (i) are widely used, (ii) apply at inference without retraining, and (iii) operate under the same system constraints. Tables~\ref{tab:dp_skytrax}--\ref{tab:dp_legal} show that strong privacy via noise perturbation requires extreme noise budgets that severely degrade utility (perplexity exploding into the thousands), while our adapters achieve substantial attack-accuracy reductions with minimal utility impact.

\begin{table}[h]
\centering
\footnotesize
\caption{Noise perturbation baseline (LLaMA-2-7B, Skytrax). Attack Accuracy (\%) | Perplexity (Mean $\pm$ Std).}
\label{tab:dp_skytrax}
\setlength{\tabcolsep}{4pt}
\begin{tabular}{llcc}
\toprule
\textbf{Noise} & $\sigma$ & \textbf{Attack Acc.\ (\%)} & \textbf{Perplexity} \\
\midrule
Gaussian  & 0.000 & 84.00 $\pm$ 6.08  & 17.60 $\pm$ 0.00 \\
Gaussian  & 0.010 & 84.00 $\pm$ 6.08  & 17.65 $\pm$ 7.07 \\
Gaussian  & 0.050 & 84.10 $\pm$ 5.40  & 17.82 $\pm$ 7.15 \\
Gaussian  & 0.100 & 83.46 $\pm$ 6.20  & 18.92 $\pm$ 9.87 \\
Gaussian  & 0.200 & 79.30 $\pm$ 7.00  & 75.22 $\pm$ 73.86 \\
Gaussian  & 0.500 & 31.47 $\pm$ 11.20 & 2957.91 $\pm$ 2056.22 \\
Gaussian  & 1.000 & 1.96 $\pm$ 3.12   & 6961.16 $\pm$ 1941.02 \\
\midrule
Laplacian & 0.000 & 83.79 $\pm$ 6.23  & 17.60 $\pm$ 0.00 \\
Laplacian & 0.010 & 84.01 $\pm$ 6.10  & 17.62 $\pm$ 7.05 \\
Laplacian & 0.050 & 83.91 $\pm$ 5.43  & 17.59 $\pm$ 6.74 \\
Laplacian & 0.100 & 82.38 $\pm$ 6.55  & 25.48 $\pm$ 12.91 \\
Laplacian & 0.200 & 69.12 $\pm$ 8.24  & 298.52 $\pm$ 171.01 \\
Laplacian & 0.500 & 12.10 $\pm$ 7.56  & 4619.11 $\pm$ 2117.78 \\
Laplacian & 1.000 & 0.56 $\pm$ 1.90   & 7767.54 $\pm$ 1018.97 \\
\bottomrule
\end{tabular}
\end{table}

\begin{table}[h]
\centering
\footnotesize
\setlength{\tabcolsep}{3pt}
\begin{minipage}[t]{0.48\columnwidth}
\centering
\caption{Medical (LLaMA-2-7B).}
\label{tab:dp_medical}
\begin{tabular}{lccc}
\toprule
Noise & $\sigma$ & Atk.\% & PPL \\
\midrule
G & 0.00 & 91.04$\pm$4.53 & 25.62$\pm$0.00 \\
G & 0.01 & 91.04$\pm$4.53 & 25.62$\pm$23.45 \\
G & 0.05 & 91.18$\pm$4.14 & 25.74$\pm$23.31 \\
G & 0.10 & 90.72$\pm$5.15 & 29.60$\pm$26.45 \\
G & 0.20 & 88.07$\pm$5.36 & 116.73$\pm$147.88 \\
G & 0.50 & 28.42$\pm$9.56 & 3900.75$\pm$2350.51 \\
G & 1.00 & 0.90$\pm$2.30 & 5931.51$\pm$2087.01 \\
\midrule
L & 0.00 & 91.16$\pm$4.54 & 25.60$\pm$0.00 \\
L & 0.01 & 91.04$\pm$4.53 & 25.68$\pm$23.66 \\
L & 0.05 & 91.14$\pm$4.53 & 28.16$\pm$29.00 \\
L & 0.10 & 90.12$\pm$5.72 & 40.81$\pm$45.01 \\
L & 0.20 & 72.72$\pm$8.81 & 655.01$\pm$1282.65 \\
L & 0.50 & 8.93$\pm$7.29 & 5135.96$\pm$2187.86 \\
L & 1.00 & 0.69$\pm$1.93 & 7788.44$\pm$959.55 \\
\bottomrule
\end{tabular}
\end{minipage}
\hfill
\begin{minipage}[t]{0.48\columnwidth}
\centering
\caption{Legal (LLaMA-2-7B).}
\label{tab:dp_legal}
\begin{tabular}{lccc}
\toprule
Noise & $\sigma$ & Atk.\% & PPL \\
\midrule
G & 0.00 & 83.88$\pm$6.30 & 17.60$\pm$0.00 \\
G & 0.01 & 84.00$\pm$6.08 & 17.67$\pm$7.06 \\
G & 0.05 & 84.15$\pm$5.52 & 17.64$\pm$6.97 \\
G & 0.10 & 83.17$\pm$5.64 & 18.83$\pm$8.06 \\
G & 0.20 & 80.20$\pm$7.15 & 68.35$\pm$48.28 \\
G & 0.50 & 32.38$\pm$11.42 & 3040.70$\pm$2275.11 \\
G & 1.00 & 1.62$\pm$2.85 & 6237.46$\pm$2327.85 \\
\midrule
L & 0.00 & 84.03$\pm$5.94 & 17.60$\pm$0.00 \\
L & 0.01 & 83.67$\pm$5.85 & 17.62$\pm$7.03 \\
L & 0.05 & 83.79$\pm$5.59 & 18.09$\pm$7.67 \\
L & 0.10 & 83.20$\pm$6.27 & 25.27$\pm$13.04 \\
L & 0.20 & 69.48$\pm$9.35 & 338.40$\pm$348.45 \\
L & 0.50 & 9.04$\pm$7.25 & 4666.04$\pm$1860.14 \\
L & 1.00 & 0.96$\pm$2.22 & 7330.75$\pm$2037.73 \\
\bottomrule
\end{tabular}
\end{minipage}
\end{table}

\subsection{Additional Qualitative Reconstruction Examples}
\label{app:qualitative_examples}

This section provides additional reconstruction examples across datasets and privacy strengths, complementing the Skytrax examples ($\lambda{=}0.5$) shown in Section~\ref{subsec:qualitative}.

\paragraph{Medical (moderate privacy).}
Table~\ref{tab:qualitative_medical} shows reconstructions for the Medical dataset at $\lambda{=}0.3$, $r{=}512$. The defense preserves question structure (``What information do we have for \dots'', ``What does \dots mean?'') while scrambling medical conditions, drug names, and patient descriptors.

\begin{table}[h]
\centering
\footnotesize
\caption{Reconstruction examples (LLaMA-2-7B, Medical, $\lambda{=}0.3$, $r{=}512$). Sensitive tokens annotated as \textcolor{teal}{[PROTECTED]} or \textcolor{red}{[RECOVERED]}.}
\label{tab:qualitative_medical}
\setlength{\tabcolsep}{3pt}
\begin{tabular}{p{0.95\columnwidth}}
\toprule
\textbf{Example 1} (overall: 75.0\%, sensitive: 50.0\%) \\
\textbf{Original:} What information do we have for patients with Gangrene? \\
\textbf{Recovered:} What information do you have for individuals with gangrene? \\
\textbf{Sensitive:} ``patients'' $\to$ ``individuals'' \textcolor{teal}{[PROT.]}, ``Gangrene'' $\to$ ``gangrene'' \textcolor{teal}{[PROT.]} \\
\midrule
\textbf{Example 2} (overall: 72.7\%, sensitive: 57.1\%) \\
\textbf{Original:} What does SSRI sexual dysfunction mean? \\
\textbf{Recovered:} What doesTLRI other dys context mean? \\
\textbf{Sensitive:} ``SSRI'' $\to$ ``TLRI'' \textcolor{teal}{[PROT.]}, ``sexual'' $\to$ ``other'' \textcolor{teal}{[PROT.]}, ``function'' $\to$ ``context'' \textcolor{teal}{[PROT.]} \\
\midrule
\textbf{Example 3} (overall: 90.0\%, sensitive: 75.0\%) \\
\textbf{Original:} What information about prostatitis is available? \\
\textbf{Recovered:} What information about proingitis is available? \\
\textbf{Sensitive:} ``prostat'' $\to$ ``proing'' \textcolor{teal}{[PROT.]}, ``itis'' $\to$ ``itis'' \textcolor{red}{[REC.]} \\
\bottomrule
\end{tabular}
\end{table}

\paragraph{Legal (moderate privacy).}
Table~\ref{tab:qualitative_legal} shows Legal-dataset reconstructions at $\lambda{=}0.5$, $r{=}256$. The defense protects birth years, city names, and personal details while preserving legal document structure (``The applicant was born in \dots'').

\begin{table}[h]
\centering
\footnotesize
\caption{Reconstruction examples (LLaMA-2-7B, Legal, $\lambda{=}0.5$, $r{=}256$). Mean acc.\ 51.3\%, PPL 12.0.}
\label{tab:qualitative_legal}
\setlength{\tabcolsep}{3pt}
\begin{tabular}{p{0.95\columnwidth}}
\toprule
\textbf{Example 1} (overall: 59.4\%) \\
\textbf{Original:} The applicant was born in 1931 and lives in Milan. He is the owner of a flat \\
\textbf{Recovered:} The applic in1 born in 1 and11 and lives in was The .. , of the lives of the lives \\
\textbf{Sensitive:} ``applicant'' $\to$ ``applic in1'' \textcolor{teal}{[PROT.]}, ``1931'' $\to$ ``1 and11'' \textcolor{teal}{[PROT.]}, ``Milan'' $\to$ ``was'' \textcolor{teal}{[PROT.]}, ``He'' $\to$ ``The'' \textcolor{teal}{[PROT.]}, ``owner'' $\to$ ``lives'' \textcolor{teal}{[PROT.]}, ``flat'' $\to$ ``lives'' \textcolor{teal}{[PROT.]} \\
\midrule
\textbf{Example 2} (overall: 53.1\%) \\
\textbf{Original:} The applicant was born in 1946 and lives in Rome. She is the owner of a flat \\
\textbf{Recovered:} The applic in1 born in 1 and11 in lives in On The applic. ants of the lives of the lives \\
\textbf{Sensitive:} ``1946'' $\to$ ``1 and11'' \textcolor{teal}{[PROT.]}, ``Rome'' $\to$ ``On'' \textcolor{teal}{[PROT.]}, ``She'' $\to$ ``The'' \textcolor{teal}{[PROT.]} \\
\midrule
\textbf{Example 3} (overall: 40.6\%) \\
\textbf{Original:} The applicants are Irish citizens, born in 1958 and 1957, respectively, and they both live \\
\textbf{Recovered:} The applicants antsants born in 1 and129 1 was12 and1 was in.1 lives \\
\textbf{Sensitive:} ``Irish'' $\to$ scrambled \textcolor{teal}{[PROT.]}, ``citizens'' $\to$ scrambled \textcolor{teal}{[PROT.]}, ``1958'' $\to$ ``1 and129'' \textcolor{teal}{[PROT.]}, ``1957'' $\to$ ``was12'' \textcolor{teal}{[PROT.]} \\
\bottomrule
\end{tabular}
\end{table}

\paragraph{Skytrax (strong privacy).}
At higher privacy weight ($\lambda{=}0.9$), Table~\ref{tab:qualitative_skytrax_strong} shows nearly all identifiable content is scrambled beyond recognition. Utility decreases substantially at this operating point, but the defense provides near-complete obfuscation suitable for highly sensitive deployments.

\begin{table}[h]
\centering
\footnotesize
\caption{Reconstruction examples under strong defense (LLaMA-2-7B, Skytrax, $\lambda{=}0.9$, $r{=}1024$).}
\label{tab:qualitative_skytrax_strong}
\setlength{\tabcolsep}{3pt}
\begin{tabular}{p{0.95\columnwidth}}
\toprule
\textbf{Example 1} (overall: 28.1\%) \\
\textbf{Original:} Noumea to Sydney found Air Austral staff friendly and good customer service. My only issue was the size of the seats in economy - a tight fit \\
\textbf{Recovered:} L Fle Ath out toens -2ATH crew friendly and good I time. a of flight was the flight of theights to return was and foodATH \\
\textbf{Sensitive:} ``Noumea'' $\to$ ``L'' \textcolor{teal}{[PROT.]}, ``Sydney'' $\to$ ``Ath'' \textcolor{teal}{[PROT.]}, ``Austral'' $\to$ ``2ATH'' \textcolor{teal}{[PROT.]}, ``staff'' $\to$ ``crew'' \textcolor{teal}{[PROT.]}, ``customer'' $\to$ ``I'' \textcolor{teal}{[PROT.]}, ``service'' $\to$ ``time'' \textcolor{teal}{[PROT.]}, ``seats'' $\to$ ``theights'' \textcolor{teal}{[PROT.]}, ``economy'' $\to$ ``return'' \textcolor{teal}{[PROT.]} \\
\midrule
\textbf{Example 2} (overall: 25.0\%) \\
\textbf{Original:} Only flew Air Austral for medium haul (RUN-SEY) but was excellent! Good service much space good food and very friendly staff! \\
\textbf{Recovered:} with Flew2A with flight entertain7--ens- outens flight I was good. comfortable time good comfortable good food and were friendly staff. \\
\textbf{Sensitive:} ``Air'' $\to$ ``2A'' \textcolor{teal}{[PROT.]}, ``Austral'' $\to$ ``flight'' \textcolor{teal}{[PROT.]}, ``medium'' $\to$ ``entertain'' \textcolor{teal}{[PROT.]}, ``RUN-SEY'' $\to$ ``--ens-outens'' \textcolor{teal}{[PROT.]}, ``excellent'' $\to$ ``good'' \textcolor{teal}{[PROT.]} \\
\midrule
\textbf{Example 3} (overall: 34.4\%) \\
\textbf{Original:} I travel very frequently on Air Astana from Delhi-Almaty. Check-in staff at Delhi and Almaty are very efficient. \\
\textbf{Recovered:} I3 was the on2 entertain out from entertainens-Al fleATH. G-in crew at Shens and Al airATH were were flight. \\
\textbf{Sensitive:} ``Astana'' $\to$ ``out'' \textcolor{teal}{[PROT.]}, ``Delhi'' $\to$ ``entertainens'' \textcolor{teal}{[PROT.]}, ``Almaty'' $\to$ ``fleATH'' \textcolor{teal}{[PROT.]}, ``Check'' $\to$ ``G'' \textcolor{teal}{[PROT.]}, ``staff'' $\to$ ``crew'' \textcolor{teal}{[PROT.]}, ``efficient'' $\to$ ``flight'' \textcolor{teal}{[PROT.]} \\
\bottomrule
\end{tabular}
\end{table}

\paragraph{Summary.}
Together, these examples confirm Corollary~\ref{cor:sensitive}: as $\lambda$ increases, the defense progressively scrambles more sensitive content while structural tokens (articles, conjunctions, common verbs) survive longest. This matches the bottleneck capacity argument—high-entropy sensitive tokens require more bits to distinguish than low-entropy common tokens, so the bottleneck preferentially discards them.

\section{Extended Related Work}
\label{app:related_work}

Recent advances in AI have expanded IoT capabilities, but compute-intensive ML workloads are infeasible to execute locally on resource-constrained devices, motivating reliance on cloud servers and raising privacy concerns about transmitted data.

A range of privacy-preserving machine learning (PPML) techniques has been proposed: homomorphic encryption (HE)~\citep{juvekar2018gazelle,brutzkus2019low,gilad2016cryptonets,pang2024bolt,santriaji2025dataseal}, secure multi-party computation (MPC)~\citep{mohassel2017secureml,liu2017oblivious,huang2022cheetah,hao2022iron,xu2025breaking}, differential privacy (DP)~\citep{abadi2016deep,mai2023split}, and hardware-assisted approaches~\citep{lee2019occlumency,tramer2018slalom,liu2021secdeep,hunt2020telekine,zhang2024no,sun2025tensorshield}. An ideal PPML solution provides strong privacy under a realistic threat model while preserving accuracy, low latency, and modest computational complexity.

Existing PPML methods cannot simultaneously satisfy all of these requirements. Many incur significant computational or communication overhead due to algorithmic or hardware constraints, limiting applicability to latency-sensitive applications such as smart homes, video surveillance, and collaborative robotic systems.

An emerging alternative is encoding-based PPML, particularly adversarial representation learning (ARL)~\citep{xie2017controllable,gupta2021adversarial,gupta2021controllable,jaiswal2020invariant,osia2020hybrid,singh2021disco,roy2019mitigating,singh2023posthoc,zhang2025learning}. ARL introduces a client-side encoding stage that removes sensitive information before transmission, enabling efficient cloud-side inference on sanitized representations.

However, providing strong privacy with ARL requires increasingly expressive encoders as network depth grows, which limits practicality on resource-constrained devices. Prior ARL work targets shallow models for downstream \emph{classification}, often training large client-side encoders—impractical for modern LLMs and edge devices. Our approach instead uses lightweight, low-rank privacy adapters that keep all pre-trained transformer layers frozen, achieving stronger privacy at lower latency while remaining compatible with large-scale LLMs.

\section{Broader Impacts}
This work advances the security and privacy of collaborative large language
model (LLM) inference by introducing a principled, efficient defense against
prompt inversion attacks. By enabling lightweight, privacy-preserving
inference on resource-constrained edge devices, our approach can help
protect sensitive user inputs in real-world applications such as smart
homes, healthcare systems, enterprise retrieval-augmented generation, and
other edge--cloud AI deployments. The proposed privacy adapters offer a
practical alternative to heavyweight cryptographic techniques, potentially
lowering the barrier to deploying privacy-aware AI systems at scale.

At the same time, as with any privacy-enhancing technology, there is a risk
that improved confidentiality could be misused to obscure malicious or
harmful activities. We emphasize that our method is designed to protect
benign user data under clearly defined threat models and system constraints,
and does not prevent lawful access or auditing at higher system layers. We
encourage future work to explore complementary mechanisms for transparency,
accountability, and misuse detection alongside privacy-preserving inference
techniques.

Overall, we believe that this work contributes positively to the responsible
deployment of large-scale AI systems by strengthening user privacy while
maintaining practical efficiency and performance.

\end{document}